\documentclass[11pt,letterpaper]{article}

\usepackage[numbers]{natbib}
\usepackage{xspace}
\usepackage{xcolor}
\usepackage{graphicx}
\usepackage{amsmath,amssymb,amsthm}
\usepackage[margin=1in]{geometry}

\usepackage{todonotes}

\theoremstyle{plain}
\newtheorem{theorem}{Theorem}[section]
\newtheorem{main-result}{Main Result}
\newtheorem{proposition}[theorem]{Proposition}
\newtheorem{lemma}[theorem]{Lemma}

\theoremstyle{definition}
\newtheorem{definition}[theorem]{Definition}
\newtheorem{example}[theorem]{Example}

\newcommand{\growingmid}{\mathrel{}\middle|\mathrel{}}

\renewcommand{\Pr}[1]{\mbox{\rm\bf Pr}\left[#1\right]}
\newcommand{\Ex}[1]{\mbox{\rm\bf E}\left[#1\right]}

\newenvironment{proof*}[1][Proof]{\medskip\noindent\textit{#1.} }{}

\begin{document}

\title{Best-Response Dynamics in Combinatorial Auctions\\with Item Bidding\thanks{An extended abstract appeared in Proceedings of the 28th ACM-SIAM Symposium on Discrete Algorithms, SODA 2017, Barcelona, Spain.}}

\author{
  Paul D\"utting\thanks{Department of Mathematics, London School of Economics,
  Houghton Street, London WC2A 2AE, UK
  (\texttt{p.d.duetting@lse.ac.uk}). Part of this work was done while the author was Senior Researcher in the Department of Computer Science at ETH Z\"urich.}
  \and
  Thomas Kesselheim\thanks{Department of Computer Science, 
  TU Dortmund, Otto-Hahn-Str.~14, 44221 Dortmund, Germany 
  (\texttt{thomas.kesselheim@cs.tu-dortmund.de}). This work was done 
  while the author was at Max Planck Institute for Informatics and Saarland 
  University, supported in part by the DFG through Cluster of Excellence MMCI.}
}

\date{\today}

\maketitle

\begin{abstract}
In a combinatorial auction with item bidding, agents participate in multiple single-item second-price auctions at once. As some items might be substitutes, agents need to strategize in order to maximize their utilities. A number of results indicate that high welfare can be achieved this way, giving bounds on the welfare at equilibrium. Recently, however, criticism has been raised that equilibria are hard to compute and therefore unlikely to be attained.

In this paper, we take a different perspective. We study simple best-response dynamics. That is, agents are activated one after the other and each activated agent updates his strategy myopically to a best response against the other agents' current strategies. Often these dynamics may take exponentially long before they converge or they may not converge at all. However, as we show, convergence is not even necessary for good welfare guarantees. Given that agents' bid updates are aggressive enough but not too aggressive, the game will remain in states of good welfare after each agent has updated his bid at least once.

In more detail, we show that if agents have fractionally subadditive valuations, natural dynamics reach and remain in a state that provides a $1/3$ approximation to the optimal welfare after each agent has updated his bid at least once. For subadditive valuations, we can guarantee an $\Omega(1/\log m)$ approximation in case of $m$ items that applies after each agent has updated his bid at least once and at any point after that. The latter bound is complemented by a negative result, showing that no kind of best-response dynamics can guarantee more than an $o(\log \log m/\log m)$ fraction of the optimal social welfare.
\end{abstract}


\section{Introduction}

In a combinatorial auction, $n$ players compete for the assignment of $m$ items. The players have private preferences over bundles of items as expressed by a valuation function $v_i\colon 2^{[m]} \rightarrow \mathbb{R}_{\ge 0}$. Our goal in this work is to find a partition of the items into sets $S_1, \ldots, S_n$ that maximizes social welfare $\sum_{i} v_i(S_i)$, based on reported valuations (bids) $b_i\colon 2^{[m]} \rightarrow \mathbb{R}_{\ge 0}$ with the freedom to impose payments $p_1, \dots, p_n$ on the players.

Even if valuations are known, finding an allocation that maximizes social welfare is typically $\mathsf{NP}$-hard. Furthermore, since valuations are assumed to be private information, some mechanics are needed to extract this information. The traditional approach is to incentivize players to bid truthfully. Insisting on truthfulness has the advantage that for the individual players it is easy to participate as it is not necessary to act strategically. However, truthfulness requires central coordination of the entire allocation and payments.

An alternative approach to this problem that is arguably seen more often in practice is to let players participate in a simpler, non-truthful mechanism and to accept strategic behavior. To derive theoretical performance guarantees, one
then seeks to prove bounds on the so-called Price of Anarchy, the worst-case ratio between the optimal social welfare and the welfare at equilibrium. The most prominent example in the context of combinatorial auctions is \emph{item bidding}, where the items are sold through separate single-item auctions. 

One can show that for pretty general classes of valuations, such as submodular or the even more general classes fractionally subadditive and subadditive, all equilibria from a broad range of equilibrium concepts obtain a decent fraction of the optimal social welfare. More recently, however, these results have been criticized for ignoring the computational complexity of finding an equilibrium. In fact, by now, there is quite a selection of impossibility results showing that finding exact equilibria is often computationally intractable. 

Our approach in this paper is different. We consider simple, best-response dynamics, in which players are activated in a round-robin fashion and players when activated buy their favorite set of items at the current prices, in a myopic way. Christodoulou et al.~\cite{ChristodoulouKS16} showed that one instance of such dynamics converges if players' valuation functions are fractionally subadditive. However, they also showed that it takes exponential time. For subadditive valuations, even convergence cannot be guaranteed because any fixed point would be a pure Nash equilibrium, and pure Nash equilibria may not exist (see Appendix \ref{app:no-pne}). We show that despite possibly long convergence time or no convergence at all, the social welfare reaches a good level very fast.


\subsection{The Setting}

We study combinatorial auctions with $n$ bidders $N$ and $m$ items $M$. Each bidder $i \in N$ has a valuation function $v_i\colon2^M \rightarrow \mathbb{R}_{\ge 0}$. Our objective is to find a feasible allocation, i.e., a partition of the items, $S_1, \ldots, S_n$, that maximizes social welfare $\sum_{i \in N} v_i(S_i)$. We assume that an allocation of items to bidders is found by distributed strategic behavior of the bidders using item bidding. That is, each bidder $i \in N$ places a bid $b_{i,j}$ on each item $j \in M$. Each item $j \in M$ is assigned to the bidder $i \in N$ with the highest bid $b_{i,j}$ at a price of $p_j = \max_{i' \neq i} b_{i',j}$. Ties are broken in an arbitrary, but fixed manner. 

We assume that bidders choose their bids strategically so as to maximize their quasi-linear utilities. Bidder $i$'s utility $u_i$ as a function of the bids $b = (b_{i'})_{i' \in N}$ is $u_i(b) = v_i(S) - \sum_{j \in S} p_j$, where $S$ is the set of items won by bidder $i$.

We say that a bid $b_i$ is a best response to the bids $b_{-i}$ if bidder $i$'s utility is maximized by $b_i$. That is, $u_i(b_i,b_{-i}) \geq u_i(b'_i,b_{-i})$ for all $b'_i$. Note that any best response must give bidder $i$ a set of items $S$ that maximizes $u_i(b) = v_i(S) - \sum_{j \in S} p_j$. We call these sets of items demand sets. A (pure) Nash equilibrium in this setting is a profile of bids $b = (b_{i'})_{i' \in N}$ such that for each bidder $i \in N$ his bid $b_i$ is a best response against bids $b_{-i}$.

We study simple game-playing dynamics in which bidders get activated in turn and myopically choose to play a best response. More formally, starting from an initial bid vector $b^0$, in each time step $t \geq 1$, some bidder $i \in N$ is activated and updates his bid $b_{i}^{t-1}$ from the previous round to a best response to the  other players' bids $b^{t}_{-i} = b^{t-1}_{-i}$ which do not change from the previous to the current round.  The fixed points of such best-response dynamics are Nash equilibria. However, Nash equilibria do not necessarily exist and even if they do best-response dynamics may not converge.

We will evaluate best-response dynamics by the social welfare that they achieve. For bid profile $b$ and corresponding allocation $S_1, \dots, S_n$ we write $SW(b) = \sum_{i} v_i(S_i)$ for the social welfare at bid profile $b$. We seek to compare this to the optimal social welfare $OPT(v)$.


\subsection{Variants of Best-Response Dynamics}

Since payments in combinatorial auctions with item bidding are second price, there are typically many ways to choose a best response. Clearly, not all best responses will ensure that good states (in terms of social welfare) will be reached quickly.

\begin{example}[Gross Underbidding]
Consider a single-item auction with $n$ bidders. Suppose $v_1 = C$ and $v_i = 1$ for $i \geq 2$, where $C \gg 1$. Suppose we start at $b=(0,\dots,0)$ and the item assigned to bidder $1$. A possible best response sequence has bidders update their bids in round-robin fashion, each time increasing the winning bid by $\epsilon$.
\end{example}

\begin{example}[Gross Overbidding]
Consider the same setting as in the previous example. If in the first round of updates the last bidder bids $C+\epsilon$ this will terminate the dynamics.
\end{example}

Note that in both these examples the social welfare after each round of best responses (and on average) is $1$, which can be arbitrarily smaller than the optimal social welfare $C$.

The issue in each of these examples is as follows. Through the bids $b_{i,j}$, the bidders effectively declare additive valuations. The allocation maximizes the declared welfare $DW(b) = \sum_{i} \sum_{j \in S_i} b_{i,j}$, which usually differs from the actual welfare $SW(b)$. In both examples, there exist update steps in which the declared utility of the bidder, i.e., $u_i^D(b) = \sum_{j \in S_i} b_{i,j} - \sum_{j \in S_i} \max_{k \neq i} b_{k,j}$,  is very different from his actual utility. We will prove bounds on the welfare achieved by best-response dynamics that are quantified by the extent to which declared utilities can differ from the actual utilities as captured by the following definitions.

\begin{definition}\label{def:aggressive}
Let $\alpha \geq 0$. We call a bid $b_i$ by bidder $i$ against bids $b_{-i}$ \emph{$\alpha$-aggressive} if $u_i^D(b) \geq \alpha \cdot \max_{b_i'} u_i(b_i', b_{-i})$.
\end{definition}

\begin{definition}\label{def:safe}
Let $\beta \geq 1$. A bidding dynamic is \emph{$\beta$-safe} if it ensures that $u_i^D(b) \leq \beta \cdot u_i(b)$ for all players $i$ and reachable bid profiles $b$.
\end{definition}

We will usually apply Definition \ref{def:aggressive} when $b_i$ is a best response to $b_{-i}$. 
However, it also leaves the freedom to consider approximate best responses.
We will see that one way to achieve Definition \ref{def:safe} is to require strong no overbidding, but we will also see an example of safe dynamics that allow overbidding.
Note that in both cases  players will have non-negative actual utilities at all times because $u_i(b^t) \geq \frac{1}{\beta} \cdot u_i^D(b^t) \geq 0$ for every bidder $i$ and time step $t$. 


\subsection{Our Results}

Our first main result is that round-robin best-response dynamics are capable of reaching states with near-optimal social welfare strikingly fast, despite the fact that convergence to equilibrium may take exponentially long or they may not converge at all.

In fact, our result applies to any round-robin bidding dynamics, provided that players choose bids that are aggressive enough but not too aggressive. This, in particular, includes dynamics in which players choose to play only approximate best responses. Also, their way of making choices does not need to be consistent in any way.

\begin{main-result}\label{main-result-1}
In a $\beta$-safe round-robin bidding dynamic with $\alpha$-aggressive bid updates the social welfare at any time step $t \geq n$ satisfies
\[
	SW(b^t) \geq \frac{\alpha}{(1+\alpha+\beta)\beta} \cdot OPT(v).
\]
\end{main-result}

In other words, once every player had the chance to update his bid, the social welfare, at any time step after that, will be within $\alpha/(1+\alpha+\beta)\beta$ of optimal. 

For fractionally subadditive valuations and subadditive valuations there exist round-robin best-response dynamics with $(\alpha,\beta) = (1, 1)$ and $(\alpha,\beta) = (1/\ln m,1)$ respectively. 
The result for XOS requires access to demand and XOS oracles \cite{DobzinskiS06}, the result for subadditive valuations requires access to demand oracles and that the greedy algorithm for set cover problems can be executed \cite{Feige06,BhawalkarR11}.

Our guarantee on the social welfare achieved by best-response dynamics shows that these dynamics provide a $1/3$ (resp.~$\Omega(1/\log m)$) approximation to the optimal social welfare that applies after a single round of bid updates, and at any time step after that.

We also prove a bound on the average social welfare of $1/2(2+\alpha)\beta$, which improves upon the above bound for large $\beta$. In particular, for subadditive valuations it is also possible to achieve $(\alpha,\beta) = (1,\ln m)$. While the point-wise guarantee of this dynamics is only $\Omega(1/\log^2 m)$, its average social welfare is within $\Omega(1/\log m)$ of optimal.

We show that the point-wise welfare guarantee of $1/3$ for fractionally subadditive valuations is tight for the respective mechanism.
Our second main result is that the $\Omega(1/\log m)$ bounds are almost best possible in a more general sense.

\begin{main-result}\label{main-result-2}
For players with subadditive valuations no best-response dynamics in which players do not overbid on the grand bundle can guarantee a better than $o(\log\log m / \log m)$ fraction of the optimal social welfare at any time step.
\end{main-result}

For round-robin bidding dynamics, this point-wise impossibility result extends to an impossibility for the average social welfare that can be achieved.

The assumption that players do not overbid on the grand bundle is quite natural, and is satisfied by all dynamics that have been proposed in the literature. It obviously applies to strong no-overbidding dynamics, but it also applies to dynamics in which players use weak no-overbidding strategies on the items that they win and bid zero on all other items.

Our proof of the lower bound is based on a non-trivial construction exploiting the algebraic properties of linearly independent vector spaces. It presents an interesting separation from the Price of Anarchy literature, where no such lower bound can be proved.

Finally, we explore to which extent our positive results depend on round-robin activation. We show that our positive results extend to the case where at each step a player is chosen uniformly at random, while the social welfare can be as low as $O(1/n)$ of optimal when the order of activation is chosen adversarially.


\subsection{Related Work}
\label{sec:related_work}

Best-response dynamics are a central topic in Algorithmic Game Theory. Probably, the best-studied application are congestion games, where best-response dynamics always converge but, except in special cases, take worst-case exponential time before they do so \cite{Rosenthal73,MondererS96,AckermannRV08}. On the other hand, a number of results show that certain types of best-response dynamics reach states of low social cost quickly \cite{GoemansMV05,ChristodoulouMS12,BiloFFM11,FarzadOV08,Roughgarden15}. Some of these results extend to weighted congestion games, where equilibria may not exist and best-response sequences may not converge for this reason.

The study of the Price of Anarchy in combinatorial auctions with item bidding was initiated by Christodoulou et al.~\cite{ChristodoulouKS16}, and subsequently refined and improved upon in \cite{BhawalkarR11,HassidimKMN11,SyrgkanisT13,DuettingHS13,FeldmanFGL13}. Some of these bounds are based on mechanism smoothness, others are not. They provide welfare guarantees for a broad range of equilibrium concepts ranging from pure Nash equilibria, over (coarse) correlated equilibria, to Bayes-Nash equilibria. For fractionally subadditive valuations there is a smoothness-based proof that shows that the Price of Anarchy with respect to pure Nash equilibria is at most $2$ \cite{ChristodoulouKS16,SyrgkanisT13}. For subadditive valuations the Price of Anarchy with respect to pure Nash equilibria is also at most $2$ \cite{BhawalkarR11}, but the best smoothness-based proof gives a bound of $O(\log m)$ \cite{BhawalkarR11,SyrgkanisT13}.  In fact, as shown by Roughgarden~\cite{Roughgarden14}, combinatorial auctions with item bidding achieve (near-)optimal Price of Anarchy among a broad class of ``simple'' mechanisms.

Also relevant to our analysis in this context is that Christodoulou et al.~\cite{ChristodoulouKS16} gave a simple, best-response dynamics for fractionally subadditive valuations, that they called Potential Procedure. They showed that this procedure always converges to a pure Nash equilibrium, but also that it may take exponentially many steps before it converges.

Lately, attempts at proving Price of Anarchy bounds for combinatorial auctions with item bidding have been criticized for not being constructive, in the sense that the computational complexity of finding an equilibrium remained open. Dobzinski et al.~\cite{DobzinskiFK15}, for example, showed that for subadditive valuations computing a pure Nash equilibrium requires exponential communication. Regarding fractionally subadditive valuations they concluded that ``if there exists an efficient algorithm that finds an equilibrium, it must use techniques that are very different from our current ones.'' Further negative findings were reported by Cai and Papadimitriou~\cite{CaiP14}, who showed that computing a Bayes-Nash equilibrium is $\mathsf{PP}$-hard. 

Most recently, Daskalakis and Syrgkanis~\cite{DaskalakisS16} considered coarse correlated equilibria. They showed that even for unit-demand players (a strict subclass of submodular) there are no polynomial-time no-regret learning algorithms for finding such equilibria, unless $\mathsf{RP} \supseteq \mathsf{NP}$, closing the last gap in the equilibrium landscape.  However, they also proposed a novel solution concept to escape the hardness trap, no-envy learning, and gave a polynomial-time no-envy learning algorithm for XOS valuations and complemented this with a proof showing that for this class of valuations every no-envy outcome recovers at least $1/2$ of the optimal social welfare.

Further relevant work comes from Devanur et al.~\cite{DevanurMSW15}, who proposed an alternative to simultaneous second-price auctions, the so-called single-bid auction. This mechanism also admits a polynomial-time no-regret learning algorithm and by a result of \cite{BravermanMW16} achieves optimal Price of Anarchy bounds within a broader class of mechanisms. 

A final point of reference are truthful mechanisms for combinatorial auctions. While no mechanism can achieve a better than $1/m^{1/2-\epsilon}$ approximation for submodular valuations with valuation queries alone \cite{DobzinskiV16}, Dobzinski~\cite{Dobzinski16} recently managed to improve a long-standing approximation guarantee of $\Omega(1/\log m)$ for submodular valuations to $\Omega(1/\sqrt{\log m})$ for fractionally subadditive valuations, requiring access to both value and demand oracles.



\section{Achieving Aggressive and Safe Bids}\label{sec:prelims}

As already discussed, best responses are generally not unique in our settings. Our positive results require that updates are \emph{aggressive} and \emph{safe}. In this section we briefly describe how to guarantee these properties for fractionally subadditive (a.k.a.~XOS) valuations and subadditive valuations. 
The missing proofs are provided in Appendix~\ref{apx:prelims}.

A valuation function is \emph{fractionally subadditive}, or \emph{XOS}, if there are values $v_{i, j}^\ell \geq 0$ such that $v_i(S) = \max_\ell \sum_{j \in S} v_{i, j}^\ell$. It is \emph{subadditive} if for all $S, T \subseteq M$, $v_i(S \cup T) \leq v_i(S) + v_i(T)$.

The dynamics that we consider approach players in round-robin fashion. When player $i$ is activated he picks a demand set $D$ at the current prices and updates his bid as described below. Note that here we assume eager updating. This assumption leads to cleaner proofs, but is not necessary as we demonstrate in Appendix~\ref{app:lazy}.


\subsection{Bid Updates for XOS Valuations}\label{sec:xos}

For XOS valuations we can update bids as described by \cite{ChristodoulouKS16}. If $D$ is the demand set chosen by player $i$, let $(v_{i,j}^\ell)_{j \in M}$ be the supporting valuation on this demand set for which $\sum_{j \in D} v_{i,j}^\ell = v_i(D)$, and set
$
	b^t_{i,j} = v^\ell_{i,j} \;\text{for $j\in D$} \;\text{and } b^t_{i,j} = 0 \text{ otherwise.}
$
Note that these update steps can be performed in polynomial time using demand and XOS oracles.

\begin{proposition}\label{prop:xos}
Starting from an initial bid vector $b^0$ satisfying strong no-overbidding, the bid updates described above lead to a sequence of bids $b^0, b^1, b^2, \dots$ that is $1$-safe and in which each update is a $1$-aggressive best response.
\end{proposition}


\subsection{Bid Updates for Subadditive Valuations}

For subadditive functions, it is generally not possible to guarantee $\alpha = 1$ and $\beta = 1$ at the same time. We describe two different, reasonable ways of bid updates.

\medskip 
\paragraph{No-Overbidding Updates} \label{sec:subadditive-1}

Given a bid vector $b_{-i}$, define $\tilde{u}_i(S,b_{-i}) = v_i(S) - \sum_{j \in S} \max_{k \neq i} b_{k,j}$. That is, $\tilde{u}_i(S,b_{-i})$ is the utility bidder $i$ can derive from buying the set $S$. Observe that $\tilde{u}_i(\,\cdot\,,b_{-i})$ is subadditive for every $b_{-i}$. Let $D$ be an inclusion-wise minimal demand set of bidder $i$ given $b_{-i}$. We can show that $\tilde{u}_i(S,b_{-i}^t) > 0$ for all $S \subseteq D$ unless $D = \emptyset$. Therefore, by \cite{BhawalkarR11} there exists an additive approximation $a_{i}$ such that (a) $\sum_{j \in D} a_{i,j} \geq 1/\ln m \cdot \tilde{u}_i(D,b_{-i}^t)$ and (b) $\sum_{j \in S} a_{i,j} \leq \tilde{u}_i(S,b_{-i}^t)$ for all $S \subseteq D$ with the property that $a_{i,j} > 0$ for all $j \in D$. We set bids
$
	b_{i,j}^t = a_{i,j} + \max_{k \neq i} b_{k,j}^t \text{ for $j \in D$ and } b_{i, j}^t = 0 \text{ otherwise.}
$
These update steps can be performed in polynomial time with a demand oracle if it is possible to compute the additive approximation, which corresponds to executing the greedy set-cover algorithm on $\tilde{u}_i(\,\cdot\, ,b_{-i}^t)$.

\begin{proposition}\label{prop:subadditive-1}
Starting from an initial bid vector $b^0$ that satisfies strong no-overbidding, the bid updates described above lead to a sequence of bids $b^0, b^1, b^2, \dots$ that is $1$-safe and in which each update is a $(1/\ln m)$-aggressive best response.
\end{proposition}


\paragraph{Aggressive Updates} \label{sec:subadditive-2}

The basic construction is the same as above except that instead of considering $a_i$ we consider $\tilde{a}_i$ such that $\tilde{a}_{i,j} = \gamma \cdot a_{i,j}$ for all items $j \in D$, where $0 < \gamma \leq \ln m$ is such that $\sum_{j \in D} a_{i,j} = 1/\gamma \cdot \tilde{u}_i(D,b_{-i}^t)$. Note that these bids satisfy: (a) $\sum_{j \in D} \tilde{a}_{i,j} = \tilde{u}_i(D,b_{-i}^t)$ and (b) $\sum_{j \in S} \tilde{a}_{i,j} \leq \gamma \cdot \tilde{u}_i(S,b_{-i}^t)$ for all $S \subseteq D$. 

\begin{proposition}\label{prop:subadditive-2}
Starting from an initial bid vector $b^0$ that satisfies strong no-overbidding, the bid updates described above lead to a sequence of bids that is $\ln m$-safe and in which each update is a $1$-aggressive best response.
\end{proposition}


\section{Welfare Guarantees}\label{sec:upper-bounds}

In this section we prove our first main result (Theorem \ref{thm:main-1}). The theorem provides a point-wise social welfare guarantee, parametrized in $\alpha$ and $\beta$, for round-robin bidding dynamics. It shows that the social welfare is high already after a single round of updates, and remains high at every single step after that.

\begin{theorem}\label{thm:main-1}
In a $\beta$-safe round-robin bidding dynamic with $\alpha$-aggressive bid updates the social welfare at any time step $t \geq n$ satisfies
\(
	SW(b^t) \geq \frac{\alpha}{(1+\alpha+\beta)\beta} \cdot OPT(v).
\)
\end{theorem}

As we have argued in Proposition \ref{prop:xos} and Proposition \ref{prop:subadditive-1} there exist round-robin best-response dynamics with $(\alpha,\beta) = (1,1)$ for fractionally subadditive valuations and $(\alpha,\beta) = (1/\ln m,1)$ for subadditive valuations. So two corollaries of our theorem are point-wise welfare guarantees of $1/3$ and $\Omega(1/\log m)$ for 
the respective mechanisms.

We also show a welfare guarantee for the average social welfare, Theorem \ref{thm:average} below, that improves upon the pointwise guarantee for large $\beta$. Note that the term $(1-\frac{n}{T})$ is $1-o(1)$ for $T \in \omega(n)$ and at least $1/2$ for $T \geq 2n$.

\begin{theorem}\label{thm:average}
In a $\beta$-safe round-robin bidding dynamic with $\alpha$-aggressive bid updates the average social welfare in the first $T$ steps satisfies
\(
 	\frac{1}{T} \sum_{t=1}^T SW(b^t) \geq \frac{\alpha}{(2 \alpha + 1) \beta} \cdot \left(1 - \frac{n}{T} \right) \cdot OPT(v).
\)
\end{theorem}

This theorem shows that the best-response dynamics described in Proposition \ref{prop:subadditive-2} with $(\alpha,\beta) = (1,\ln m)$, whose point-wise welfare guarantee is only $\Omega(1/\log^2 m)$ by Theorem \ref{thm:main-1}, guarantees an average social welfare of $\Omega(1/\log m)$.

In Appendix \ref{app:one-third} we show that the point-wise welfare guarantee of $1/3$ for fractionally subadditive valuations is tight for the respective mechanism. In Section \ref{sec:lower-bound} we show that the $\Omega(1/\log m)$ bounds are essentially best possible in a more general sense.


\subsection{Proof of Theorem \ref{thm:main-1}}

The core of our proof of the pointwise welfare guarantee are two lemmata. The first (Lemma \ref{lem:initial-low}) shows that the declared social welfare after a single round of updates is high when the initial declared welfare is low and the second (Lemma \ref{lem:initial-high}) shows that the declared welfare after a single round of updates is high when the initial declared welfare is high. To prove these lemmata we need the following auxiliary lemma.

\begin{lemma}\label{lem:aux}
Consider a sequence $b^0, \ldots, b^n$ in which bidder $i$ updates his bid in step $i$. Denote bidder $i$'s declared utility in step $i$ by $u_i^D(b^{i})$. Then, $\sum_{i = 1}^{n} u_i^D(b^i) \leq DW(b^n)$.
\end{lemma}
\begin{proof}
Consider an arbitrary bidder $i$. Bidder $i$ updates his bid in step $i$. 
Suppose bidder $i$'s update buys him the set of items $S'$. Then
\[
	u_i^D(b^{i}) = \sum_{j \in S'} \left(b^{i}_{i,j} - \max_{k \neq i} b_{k,j}^i \right)\enspace.
\]

For $i > 0$, let $z^i_j = \max_{k \leq i} b_{k, j}^i$ for all $j$. That is, $z^i_j$ is the maximum bid on item $j$ that is placed by one of the bidders $1, \ldots, i$, $z^0_j = 0$ for all $j$.

The crucial observation is that
\(
	\sum_{j \in S'} (b_{i,j}^{i} - \max_{k \neq i} b_{k,j}^i) \leq \sum_{j \in M} (z_j^i - z_j^{i-1}) \enspace.
\)
The reason is as follows. For $j \not\in S'$, we have $z_j^i \geq z_j^{i-1}$ by definition. For $j \in S'$, $b_{i, j}^i = z^i_j$ and $\max_{k \neq i} b_{k,j}^i \geq \max_{k < i} b_{k,j}^i = \max_{k < i} b_{k,j}^{i-1} = z_j^{i-1}$.

Summing over all players $i$ we obtain
\[
	\sum_{i \in N} u_i^D(b^{i}) \leq \sum_{i \in N} \sum_{j \in M} (z_j^i - z_j^{i-1})\enspace.
\]

The double sum is telescoping and $z_j^n = \max_{k} b^n_{k,j}$ and $z_j^0 = 0$ by definition. So,
\[
	\sum_{i \in N} u_i^D(b^{i}) \leq \sum_{j \in M} (z_j^n - z_j^{0}) = \sum_{j \in M} \max_{k} b^n_{k,j} = DW(b^n)\enspace,
\]
which proves the claim. 
\end{proof}

With the help of this lemma we can now prove our key lemmata.

\begin{lemma}\label{lem:initial-low}
Let $S^\ast_1, \dots, S^\ast_n$ be any feasible allocation, in which player $i$ receives items $S^\ast_i$. Consider a sequence $b^0, \ldots, b^n$ in which bidder $i$ updates his bid in step $i$ using an $\alpha$-aggressive bid. 
We have $(\alpha + 1) \cdot DW(b^n) + \alpha \cdot DW(b^0) \geq \alpha \cdot \sum_{i \in N} v_i(S^\ast_i)$.
\end{lemma}
\begin{proof}
Consider player $i$'s action in time step $i$. Instead of choosing bid $b^i_i$, he could have bought the set of items $S^\ast_i$.
As $b^i_i$ is $\alpha$-aggressive, we get 
\[
	u_i^D(b^{i}) \geq \alpha \cdot \bigg( v_i(S^\ast_i) - \sum_{j \in S^\ast_i} \max_{k \neq i} b_{k,j}^i \bigg) \enspace.
\]

Define $p_j^t = \max_i b_{i,j}^t$ for all items $j$. That is, $p_j^t$ is the maximum bid that is placed on item $j$ in bid profile $b^t$. 
We claim that for every $j \in S^\ast_i$, $\max_{k \neq i} b_{k,j}^i \leq p_j^n + p_j^0$. This is correct because if $b_{k,j}^i$ attains its maximum for $k < i$ then $\max_{k \neq i} b_{k,j}^i \leq p_j^n$ as $k$'s bid on item $j$ will not change anymore. In the other case, if $k > i$, then $\max_{k \neq i} b_{k,j}^i \leq p_j^0$ because $k$ has not yet changed the bid on item $j$. Using that both $p_j^0$ and $p_j^n$ are never negative, the bound follows.

We thus have
\[
	u_i^D(b^{i}) + \alpha \cdot \sum_{j \in S^\ast_i} (p_j^n + p_j^0) \geq \alpha \cdot v_i(S^\ast_i)\enspace.
\]

Summing this inequality over all bidders $i \in N$ yields
\[
	\sum_{i = 1}^n u_i^D(b^{i}) + \alpha \cdot \sum_{i = 1}^n \sum_{j \in S^\ast_i} (p_j^n + p_j^0) \geq \alpha \cdot \sum_{i = 1}^n v_i(S^\ast_i) \enspace.
\]

We can upper bound the first sum by $DW(b^n)$ using Lemma \ref{lem:aux}. The double sum adds up every $j \in M$ exactly once and we have $\sum_{j \in M} p_j^n = DW(b^n)$ and $\sum_{j \in M} p_j^0 = DW(b^0)$. We obtain
\[
	(\alpha + 1) \cdot DW(b^n) + \alpha \cdot DW(b^0) \geq \alpha \cdot \sum_{i = 1}^n v_i(S^\ast_i)\enspace, 
\]
as claimed.
\end{proof}

\begin{lemma}\label{lem:initial-high}
Consider a $\beta$-safe bid sequence $b^0, \ldots, b^n$ in which player $i$ changes his bid from $b^{i-1}$ to $b^i$ using an $\alpha$-aggressive bid. Then, $DW(b^n) \geq \frac{\alpha}{\beta} \cdot DW(b^0)$.
\end{lemma}

\begin{proof}
Consider an arbitrary bidder $i$ and his update from $b^{i-1}$ to $b^{i}$. Denote the set of items that bidder $i$ won under bids $b^{i-1}$ by $S_i^{i-1}$, and the set of items that he wins under bids $b^{i}$ by $S_i^{i}$. So 
\begin{align*}
u_i^D(b^{i-1}) = \sum_{j \in S_i^{i-1}} b_{i,j}^{i-1} - \sum_{j \in S_i^{i-1}} \max_{k \neq i} b_{k,j}^{i-1}  \;\; \text{and,} \;\;
u_i^D(b^{i}) = \sum_{j \in S_i^{i}} b_{i,j}^{i} - \sum_{j \in S_i^{i}} \max_{k \neq i} b_{k,j}^{i}\enspace.%
\end{align*}

Using that for all $k \neq i$ and all $j$ we have $b_{k,j}^{i-1} = b_{k,j}^{i}$ we obtain that the difference in declared welfare over all bidders between steps $i-1$ and $i$ is equal to the difference in bidder $i$'s declared utility at these time steps. Formally,
\begin{align*}
DW(b^i)
&= \sum_{j \in M\setminus S_i^{i}} \max_{k \neq i} b_{k,j}^{i-1} + \sum_{j \in S_i^{i}} b_{i,j}^{i} \displaybreak[0]\\
&= \sum_{j \in M} \max_{k \neq i} b_{k,j}^{i-1} + \sum_{j \in S_i^{i}} b_{i,j}^{i} - \sum_{j \in S_i^{i}} \max_{k \neq i} b^{i}_{k,j}\displaybreak[0]\\
&= \sum_{j \in M} \max_{k \neq i} b_{k,j}^{i-1} + u_i^D(b_i)\displaybreak[0]\\
&= \sum_{j \in M\setminus S_i^{i-1}} \max_{k \neq i} b_{k,j}^{i-1} + \sum_{j \in S_i^{i-1}} \max_{k \neq i} b_{k,j}^{i-1} + u_i^D(b_i)\displaybreak[0]\\
&= \sum_{j \in M\setminus S_i^{i-1}} \max_{k \neq i} b_{k,j}^{i-1} + \sum_{j \in S_i^{i-1}} b_{i,j}^{i-1} + u_i^D(b_i) - \sum_{j \in S_i^{i-1}} b_{i,j}^{i-1} + \sum_{j \in S_i^{i-1}} \max_{k \neq i} b_{k,j}^{i-1}\\
&= DW(b^{i-1}) + u_i^D(b_i) - u_i^D(b^{i-1})\enspace.
\end{align*}

We now extend this identity to a lower bound on $DW(b^i)$. 
Since $b^i_i$ is $\alpha$-aggressive, 
we have $u^D_i(b^i) \geq \alpha \cdot u_i(b^{i-1})$. 
Since the bidding sequence is $\beta$-safe, $u^D_i(b^t) \leq \beta \cdot u_i(b^t)$ for all $t$. So,
\begin{align*}
DW(b^i) & = DW(b^{i-1}) + u^D_i(b^i) - u^D_i(b^{i-1}) \\
& \geq DW(b^{i-1}) + u^D_i(b^i) - \beta \cdot u_i(b^{i-1}) \\
& \geq DW(b^{i-1}) + u^D_i(b^i) - \frac{\beta}{\alpha} \cdot u^D_i(b^i) \\
& = DW(b^{i-1}) - \left(\frac{\beta}{\alpha}-1\right) \cdot u^D_i(b^i) \enspace.
\end{align*}

Summing this inequality over all bidders $i \in N$ and using the telescoping sum $\sum_{i \in N} (DW(b^i) - DW(b^{i-1}) = DW(b^n) - DW(b^0)$ we obtain
\[
DW(b^n) \geq DW(b^0)- \left(\frac{\beta}{\alpha}-1\right) \sum_{i \in N} u_i^D(b^i)\enspace.
\]
Since $\alpha \leq 1$ and $\beta \geq 1$ the factor $(\beta/\alpha-1) \geq 0$. We can therefore use Lemma \ref{lem:aux} to conclude that
\[
DW(b^n) \geq DW(b^0)- \left(\frac{\beta}{\alpha}-1\right) DW(b^n) \enspace,
\]
which concludes the proof.
\end{proof}

We will use our key lemmata to show a lower bound on the declared welfare. To relate the declared welfare to the social welfare we will use the following lemma.

\begin{lemma}\label{lem:declared-vs-actual}
In a $\beta$-safe sequence of bid profiles $b^0, b^1, b^2, \dots$ for every $t \geq 0$, $DW(b^t) \leq \beta \cdot SW(b^t)$.
\end{lemma}

\begin{proof}
Consider an arbitrary time step $t$. Since the bid profile $b^t$ is $\beta$-safe we know that for the allocation $T_1, \dots, T_n$ that corresponds to $b^t$,
\begin{align*}
	\sum_i u_i^D(b^t) 
	&= \sum_i \sum_{j \in T_i} \big(b_{i,j}^t - \max_{k \neq i} b_{k,j}^t\big)\\
	&\leq \beta \cdot \sum_i u_i(b)
	= \beta \cdot \sum_i \bigg(v_i(T_i) - \sum_{j \in T_i} \max_{k \neq i} b_{k,j}^t\bigg).
\end{align*}

Rearranging this and using that $\beta \geq 1$ we obtain
\begin{align*}
	DW(b^t) 
	= \sum_i \sum_{j \in T_i} b_{i,j}^t 
	\leq \beta \cdot SW(b^t) - (\beta-1) \sum_i \sum_{j \in T_i} \max_{k \neq i} b_{k,j}^t 
	\leq \beta \cdot SW(b^t)\enspace,
\end{align*}
and the claim follows.
\end{proof}

We are now ready to prove the theorem.

\begin{proof}[Proof of Theorem \ref{thm:main-1}]
To prove the guarantee for time step $t \geq n$ consider the bid sequence of length $n+1$ from $b^{t-n}$ to $b^t$. At time steps $t-n+1$ to $t$ each bidder updates his bid exactly once. 
By the virtue of being a subsequence of a $\beta$-safe bidding sequence the sequence $b^{t-n}, \dots, b^t$ is $\beta$-safe. Moreover each bid update is $\alpha$-aggressive. 

Applying first Lemma \ref{lem:initial-high} and then Lemma \ref{lem:initial-low} with $b^t$ taking the role of $b^n$, $b^{t-n}$ taking the role of $b^0$, and setting $S^*_1, \dots, S^*_n$ to the allocation that maximizes welfare we obtain
\begin{align*}
(1+\alpha+\beta) \cdot DW(b^t)
&= (\alpha+1) \cdot DW(b^t) + \alpha \cdot \frac{\beta}{\alpha} DW(b^t)\\
&\geq (\alpha+1) \cdot DW(b^t) + \alpha \cdot DW(b^{t-n}) \\
&\geq \alpha \cdot OPT(v)\enspace.
\end{align*}

Now, by Lemma \ref{lem:declared-vs-actual}, $DW(b^t) \leq \beta \cdot SW(b^t)$. Combining this with the previous inequality yields
\[
	(1+\alpha+\beta) \cdot \beta \cdot SW(b^t) \geq \alpha \cdot OPT(v)\enspace,
\]
as claimed.
\end{proof}


\subsection{Proof of Theorem~\ref{thm:average}}

With the proof of the pointwise welfare guarantee at hand we have already done the bulk of the work for proving our guarantee regarding the average welfare. The basic idea is to sum the lower bound on the declared welfare at any given time step as provided by Lemma \ref{lem:initial-low} over all time steps to obtain a lower bound on the average declare welfare, and to turn this into a lower bound on the actual social welfare using Lemma \ref{lem:declared-vs-actual}.

\begin{proof}[Proof of Theorem~\ref{thm:average}]
We first use Lemma~\ref{lem:initial-low} to relate the declared welfare at time steps $t$ and $t-n$ to the optimal social welfare. Namely, for all $t \geq n$,
\[
 	(\alpha + 1) \cdot DW(b^t) + \alpha \cdot DW(b^{t-n}) \geq \alpha \cdot OPT(v)\enspace.
\]

Next we take the sum over all time steps $t$ and use that $DW(b^t) \geq 0$ to obtain the following lower bound on the average declared welfare
\begin{align*}
 	\frac{1}{T} \cdot \sum_{t=1}^T DW(b^t)
	&\geq \frac{1}{T} \cdot \sum_{t=n+1}^T DW(b^t) \\
	&\geq \frac{\alpha}{\alpha + 1} \cdot \frac{1}{T} \cdot \sum_{t=n+1}^T \bigg( OPT(v) - DW(b^{t-n}) \bigg)\\
	&\geq \frac{\alpha}{\alpha + 1} \cdot \frac{T - n}{T} \cdot OPT(v) - \frac{\alpha}{\alpha + 1} \cdot \frac{1}{T} \cdot \sum_{t=1}^T DW(b^t) \enspace.
\end{align*}

Solving this inequality for $\frac{1}{T} \cdot \sum_{t=1}^T DW(b^t)$ and using Lemma \ref{lem:declared-vs-actual} to lower bound $SW(b^t)$ by  $1/\beta\cdot DW(b^t)$ we obtain
\begin{align*}
 	\frac{1}{T} \cdot \sum_{t=1}^T SW(b^t) 
	\geq \frac{1}{\beta} \cdot \frac{1}{T} \cdot \sum_{t=1}^T DW(b^t)
	\geq \frac{\alpha}{(2 \alpha + 1)\beta} \cdot \frac{T - n}{T} \cdot OPT(v) \enspace,
\end{align*}
which proves the claim. 
\end{proof}


\section{Lower Bound for Subadditive CAs}\label{sec:lower-bound}

Next we show our second main result (Theorem \ref{thm:lower-bound-two-bidders}), which shows that no best-response dynamics in which bidders do not overbid on the grand bundle can achieve a point-wise welfare guarantee that is significantly better than $1/\log m$. The assumption that bidders do not overbid on the grand bundle seems quite natural, and does allow overbidding on subsets of items. It is satisfied by all dynamics that we have described in Section \ref{sec:prelims} and more generally by all dynamics that have been proposed in the literature.

\begin{theorem}\label{thm:lower-bound-two-bidders}
For every positive integer $k \in \mathbb{N}_{>0}$ there exists an instance with $n = 2$ players, $m = 2^k-1$ items, and subadditive valuations $v = (v_1,v_2)$ such that in every best-response dynamics in which players do not overbid on the grand bundle there exist infinitely many time steps $t$ at which
\[
	SW(b^t) \leq \frac{1}{\Omega\left(\frac{\log m}{\log\log m}\right)} \cdot OPT(v).
\]
\end{theorem}

To prove this theorem we show that whenever the second player has updated is bid social welfare will be low. This does not imply that the average welfare will be low as well. However, if we restrict attention to round-robin dynamics, then we can extend the construction by adding additional players after the second player that play a low-stakes game on separate items forcing the average welfare to be low as well.


\subsection{Proof of Theorem \ref{thm:lower-bound-two-bidders}}

Our proof of the lower bound is built around the following family of hard instances, with $n = 2$ players and $m = 2^k-1$ items. The valuations of the first player are based on an example that demonstrates the worst-case integrality gap for set cover linear programs (see, e.g, \cite[Example 13.4]{Vazirani2001}), and has been used in the context of combinatorial auctions with item bidding before \cite{BhawalkarR11}. The crux of our construction is in the design of the second player's valuation function, and its interplay with the valuation function of the first player.

\begin{definition}\label{def:hard-instances}
For every positive integer $k \in \mathbb{N}_{>0}$ the hard instance $\mathcal{I}_k$ consists of $n = 2$ bidders and $m = 2^k-1$ items and the following subadditive valuations:  
\begin{enumerate}
\item First bidder: Number the items from $1$ to $m$ and let $\mathbf{i}$ be a $k$-bit binary vector representing the integer $i$. Interpret $\mathbf{i}$ as a $k$-dimensional vector over  $\mathbb{F}_2$. Write $\mathbf{i} \cdot \mathbf{j}$ as the dot product of the two vectors. Let $S_i = \{j \mid \mathbf{j} \cdot \mathbf{i} = 1 \}$. Note that each such set contains $(m+1)/2$ items, and each item is contained in $(m+1)/2$ such sets. For each set of items $T \subseteq M$ let $v_1(T)$ be the minimum number of sets $S_i$ required to cover the items in $T$. 

\item Second bidder: Set $\rho = 4 \frac{k}{m}$ and $d = k - \log_2 k$. Let $\mathcal{D}$ denote the set of all $d$-dimensional subspaces of $\mathbb{F}_2^k$ excluding the zero vector. Then for any set of items $T$ let
\begin{align*}
	&v_2(T) = \rho \cdot \max_{D \in \mathcal{D}} w_D(T)\enspace,
\quad \text{where}\\
	&w_D(T) = 
	\begin{cases}
		0 & \text{for $|T| = 0$} \\
		\frac{|D|}{2}& \text{for $0 < |T \cap D| < |D|$} \\
		|D| & \text{else}
	\end{cases}
	\enspace.
\end{align*}
\end{enumerate}
\end{definition}

Note that, in the instances just described, the first player has a valuation of $v_1(M) \geq k = \log_2(m+1)$ for the grand bundle, while the second player has a maximum valuation of $\max_{T} v_2(T) = \rho \cdot |D| = \rho \cdot (2^d-1) \leq \rho \cdot 2^d = 4$ for any set of items.

To prove the theorem we first use linear algebra to derive a symmetry property of $\mathcal{D}$, which together with weak no-overbidding of the first player on the grand bundle implies the existence of a subset of items $D \in \mathcal{D}$ with low prices (Lemma \ref{lem:prices}). Intuitively, this is because the sets of items that the second player is interested in are rather small (of size about $m/\log_2 m$), and there are sufficiently many of these sets. 
We then show that every demand set of the second player under these prices includes some set of items $D' \in \mathcal{D}$ (Lemma \ref{lem:demand-set}). In the final step, we show that if the second player buys any such set $D'$, then the first player's valuation for the remaining items $M \setminus D'$ and hence the overall social welfare is at most $O(\log\log m)$ (Lemma \ref{lem:value}).

\begin{lemma}\label{lem:prices}
Let $k \in \mathbb{N}_{>0}$. Consider the hard instance $\mathcal{I}_k$. For every vector of bids $b$ such that the first player does not overbid on the grand bundle there is a $d$-dimensional subspace $D \in \mathcal{D}$ such that $\sum_{j \in D} b_{1,j} <  \rho \cdot \frac{\lvert D \rvert}{2}$.
\end{lemma}

\begin{proof}
Since the first player does not overbid on the grand bundle we have $\sum_{j \in M} b_{1, j} \leq v_1(M) = k$, so the average bids are bounded by $\frac{1}{m} \sum_{j \in M} b_{1, j} \leq \frac{k}{m}$.

Observe that the number of $d$-dimensional subspaces of $\mathbb{F}_2^k$ that contain a vector $0 \neq x \in \mathbb{F}_2^k$ is given as $\binom{k-1}{d-1}_2$, where $\binom{\,\cdot\,}{\,\cdot\,}_q$ refers to the $q$-binomial coefficient (see, e.g., \cite{Prasad2010}). So, in particular, this number is independent of $x$. Therefore,
instead of taking the average over all items $M$, we can take the average over all sets $D \in \mathcal{D}$ and take the average within such a set, i.e., $\frac{1}{m} \sum_{j \in M} b_{1, j} = \frac{1}{\lvert \mathcal{D} \rvert} \sum_{D \in \mathcal{D}} \frac{1}{\lvert D \rvert} \sum_{j \in D} b_{1, j}$.

In combination, there has to be a $D$ such that $\frac{1}{\lvert D \rvert} \sum_{j \in D} b_{1, j} \leq \frac{1}{m} \sum_{j \in M} b_{1, j} \leq \frac{k}{m}$. Since $\frac{k}{m} < \frac{\rho}{2} = 2 \frac{k}{m}$ the claim follows. 
\end{proof}

\begin{lemma}\label{lem:demand-set}
Let $k \in \mathbb{N}_{>0}$. Consider the hard instance $\mathcal{I}_k$. 
If the prices $p$ as seen by the second player are such that $\sum_{j \in D} p_j < \rho \cdot |D|/2$ for some $D \in \mathcal{D}$, then each demand set of the second player under these prices includes some $D' \in \mathcal{D}$. 
\end{lemma}

\begin{proof}
By our assumption on the sum of the prices of the items in $D$, $u(D) = v_2(D) - \sum_{j \in D} p_j = \rho \cdot w_D(D) - \sum_{j \in D} p_j > \rho \cdot \frac{\lvert D \rvert}{2}$. Now, let $S \subseteq M$ be a demand set under $v_2$. If $\lvert S \cap D' \rvert < \lvert D' \rvert$ for all $D' \in \mathcal{D}$, then we have $u(S) = v_2(S) - \sum_{j \in S} p_j \leq v_2(S) = \rho \cdot\max_{D' \in \mathcal{D}} w_{D'}(S) \leq  \rho \cdot \max_{D' \in \mathcal{D}} \frac{\lvert D' \rvert}{2} < u(D)$. This means, $S$ can only be a demand set if $\lvert S \cap D' \rvert = \lvert D' \rvert$ for some $D' \in \mathcal{D}$.
\end{proof}

\begin{lemma}\label{lem:value}
Let $k \in \mathbb{N}_{>0}$. Consider the hard instance $\mathcal{I}_k$. Then for $D' \in \mathcal{D}$ we have $v_1(M \setminus D') \leq k - d$. 
\end{lemma}
\begin{proof}
To show the bound on $v_1$, we use that $D' \cup \{ 0 \}$ is a subspace of $\mathbb{F}_2^k$ of dimension $d$. That is, any basis $x_1, \ldots, x_d$ of $D' \cup \{ 0 \}$ can be extended by $x_{d+1}, \ldots, x_k$ to a basis of $\mathbb{F}_2^k$. Let $X = (x_1, \ldots, x_k)$. This way, $X^{-1}$ is the matrix that expresses $j \in \mathbb{F}_2^k$ as a linear combination of $x_1, \ldots, x_k$. As $x_1, \ldots, x_d$ is a basis of $D' \cup \{ 0 \}$, we know that for every $j \not\in D' \cup \{ 0 \}$ the vector $X^{-1} j$ cannot be zero in all components $d + 1, \ldots, k$. This implies that the set $M \setminus D'$ can be covered by sets $S_i$ for $i$ being the rows $d + 1, \ldots, k$ of $X^{-1}$. Therefore $v_1(M \setminus D') \leq k - d$.
\end{proof}

\begin{proof}[Proof of Theorem \ref{thm:lower-bound-two-bidders}]
Any best-response dynamics has to ask every bidder infinitely often. We claim that the social welfare is $O(\log\log m)$ right after each update of the second player. Since the optimal social welfare is $\Omega(\log m)$ this shows the claim.

Let $b^t$ be a bid vector after the second player has made a move. Using Lemma~\ref{lem:prices}, we know that there is a set $D \in \mathcal{D}$ with $\sum_{j \in D} b^{t-1}_{1,j} < \rho \cdot \frac{\lvert D \rvert}{2}$. By Lemma~\ref{lem:demand-set}, the second player then buys a superset of some $D' \in \mathcal{D}$. Therefore, right after the second player has updated his bid the first player is allocated a subset of the items $M \setminus D'$. Lemma~\ref{lem:value} implies that the social welfare for this allocation is no higher than $k - d + \rho 2^d = O(\log\log m)$.
\end{proof}


\section{Beyond Round-Robin Activation}\label{sec:activation}

Our positive results make use of the fact that bidders are activated to update their bid in round-robin fashion. That is, between two activations of a bidder, each other bidder is activated exactly once. In this section, we investigate alternative activation protocols. 

\subsection{Randomized Activation}
We first show that our positive results extend to the case where at each step a random player gets to update his bid.

\begin{theorem}\label{thm:randomized-activation}
Consider a $\beta$-safe sequence of bids that is generated by choosing at each time step a player uniformly at random and letting this player update his bid to an $\alpha$-aggressive bid. Then for any time step $T \geq n$, $\Ex{SW(b^T)} \geq \frac{\alpha}{2 (1 + 4 \alpha) \beta} \cdot OPT(v)$.
\end{theorem}

The key difference to the previous positive results is as follows. In the case of round-robin activation, we could bound the price that a bidder has to pay for an item $j$ at any time by the sum of the maximum bid before the first and after the $n$-th step. As now, in the case of random activation, a bidder can potentially be activated multiple times during the first $n$ steps, this is not true anymore. Instead, we can show the following lemma.

\begin{lemma}
\label{lem:max-vs-final}
Consider a sequence of bids that is generated by choosing at each time step a player uniformly at random and letting this player update his bid. Then, for all items $j \in M$ and all lengths of the sequence $T \geq 0$, we have
\[
\Ex{\max_{t \leq T} \max_i b_{i, j}^t} \leq \left( 1 - \frac{1}{n} \right)^{-T} \Ex{\max_i b_{i, j}^T} .
\]
\end{lemma}

The proof can be found in Appendix~\ref{app:randomized-activation}. The overall idea is to bound the probability that a bidder who causes a high bid is  activated again. Using this lemma, we can follow a similar pattern as when proving Theorem~\ref{thm:main-1}.

\begin{proof}[Proof of Theorem \ref{thm:randomized-activation}]
Since all of our arguments apply starting from any vector of bids, we can without loss of generality assume that $T$ is the final of a sequence of $n$ bid updates, and so $T = n$. Let $N'$ be the set of players that are selected to bid at least once during this sequence of bid updates. Denote by $S^\ast_1, \dots, S^\ast_n$ the allocation that maximizes social welfare.
By a variant of Lemma \ref{lem:initial-low}, which does not make use of round-robin activation and is given as Lemma \ref{lem:max-low} in Appendix~\ref{app:randomized-activation}, we have
\[
 	DW(b^T) + \alpha \sum_{j \in M} \max_{t \leq T} \max_{i} b_{i,j}^t \geq \alpha \sum_{i \in N'} v_i(S^\ast_i)\enspace.
\]
Note that $DW(b^T)$, $\max_{t \leq T} \max_{i} b_{i,j}^t$, and $N'$ are now random variables. Taking expectations of both sides, we get
\[
 	\mathbf{E}\bigg[ DW(b^T) + \alpha \sum_{j \in M} \max_{t \leq T} \max_{i} b_{i,j}^t \bigg] \geq \mathbf{E}\bigg[ \alpha \sum_{i \in N'} v_i(S^\ast_i)\bigg]\enspace.
\]
By linearity of expectation, this implies
\[
 	\Ex{DW(b^T)} + \alpha \sum_{j \in M} \Ex{\max_{t \leq T} \max_{i} b_{i,j}^t} \geq \alpha \sum_{i \in N} \Pr{ i \in N' } v_i(S^\ast_i)\enspace.
\]
The probability of each player to be selected at least once is $\Pr{ i \in N' } = 1 - \left( 1 - \frac{1}{n} \right)^T$. 
%
Lemma~\ref{lem:max-vs-final} shows that $\Ex{\sum_{j \in M} \max_{t \leq T} \max_{i} b_{i,j}^t} \leq \left( 1 - \frac{1}{n} \right)^{-T} \Ex{DW(b^T)}$. 

We obtain
\[
	\bigg( 1 + \alpha \left( 1 - \frac{1}{n} \right)^{-T} \bigg) \Ex{DW(b^T)} \geq \alpha \bigg( 1 - \left( 1 - \frac{1}{n} \right)^T \bigg) \cdot \sum_{i \in N} v_i(S^\ast_i) \enspace,
\]
and therefore
\[
 	\Ex{DW(b^T)} \geq \alpha \cdot \frac{1 - \left( 1 - \frac{1}{n} \right)^T }{1 + \alpha \left( 1 - \frac{1}{n} \right)^{-T}} \cdot \sum_{i \in N} v_i(S^\ast_i)\enspace.
\]

Finally, we use Lemma \ref{lem:declared-vs-actual} to relate the declared social welfare to the actual social welfare and the fact that $T = n \geq 2$ to lower bound $1-(1-1/n)^n \geq 1/2$ and upper bound $(1-1/n)^{-n} \leq 4$. This yields,
\[
 	\Ex{SW(b^T)} 
	\geq \frac{\alpha}{2 (1 + 4 \alpha) \beta} \cdot OPT(v) \enspace. \qedhere
\]
\end{proof}

\subsection{Adversarial Activation}
We conclude by showing that our positive results that show quick convergence to states of high welfare 
no longer apply if an adversary chooses the order in which players get to update their bids.  
Our result concerns XOS valuations, and $1$-safe bidding sequences in which each bid
update is to a $1$-aggressive best response. It applies even if players 
update their bids as in the Potential Procedure of \cite{ChristodoulouKS16}. That is, unless the 
activated player already plays a best response, he chooses an arbitrary demand set and bids 
his supporting additive valuation on the respective set and zero on all other items.
\begin{theorem}\label{thm:adversarial-activation}
For every $\epsilon > 0$, $n$, and $k$, there is an instance with $n$ agents with XOS valuations and $(n-1) \cdot (k+1)$ items, 
an initial bid vector $b^0$, and an activation sequence such that, even if each activated agent
updates his bid as in the Potential Procedure, until
each agent has been activated $\Omega(2^k)$ times the welfare has
never exceeded a $\frac{1+\epsilon}{n-1}$ fraction of the optimum. 
\end{theorem}

At the core of our proof (in Appendix \ref{app:adversarial-activation}) is the following
proposition that applies even if players have unit-demand valuations, i.e.,
a player's valuation for a set of items is the maximum value for any item in the set. It
shows the existence of a cyclic activation pattern in which each player gets to 
update his bid, but the dynamic remains in states of low welfare.
The construction assumes that players also update their bid if this does not strictly improve 
their utility, and that ties among multiple best responses are broken in our favor.

\begin{proposition}
\label{prop:adversarial-activation}
For every $\epsilon > 0$ and $n$, there is an instance of $n$ agents with unit-demand valuations for $n-1$ items, an
initial bid vector $b^0$, and a cyclic activation pattern in which every agent is activated
at least once
and bid updates are as in the Potential Procedure except that updates need not be
strict improvements and ties among multiple best responses are broken in our favor,
but the social welfare is always at most a $\frac{1 + \epsilon}{n - 1}$ fraction of the optimal welfare. 
\end{proposition}

\begin{proof}
There are $n$ bidders and $n - 1$ items. Player $i$'s valuation for a set $S \subseteq M$ is given as $v_i(S) = \max_{j \in S} v_{i, j}$. For bidder $1$, we let $v_{1, 1} = \ldots, v_{1, n - 1} = 1 + \epsilon$. For bidder $i > 1$, define $v_{i, i-1} = 1$ and $v_{i, j} = 0$ for $j \neq i - 1$. The social optimum assigns item $j$ to bidder $j + 1$ and has welfare $n - 1$.

In the initial bid vector $b^0$ all players bid zero. The activation scheme is as follows: In every odd step bidder $1$ makes a move, while in even steps bidders $i > 1$ are activated in a round-robin way. That is, the activation works repeatedly as $1, 2, 1, 3, 1, 4, \ldots, 1, n-1, 1, n$.

With this activation order, it's possible that bidder $1$ bids $1+\epsilon$ on item $t$ the $t$-th time he is activated, while bidders $i > 1$, when activated, see a bid of $1+\epsilon$ on the item they are interested in, and therefore bid $0$ on all items. This way the social welfare at any time step $t \geq 1$ is $1+\epsilon$.
\end{proof}

Our proof in the appendix combines this construction with several copies of the exponential 
lower-bound construction of Theorem 3.4~in~\cite{ChristodoulouKS16}, and thus ensures that 
each update is a strict improvement and unique.


\section{Concluding Remarks and Outlook}

In our analysis we focused on fractionally subadditive and subadditive valuations, which do not exhibit complements. A natural question is whether similar results can be obtained for classes of valuations that exhibit complements. In Appendix~\ref{app:mphk}, we discuss an example with MPH-$k$ valuations \cite{FeigeFIILS15} that highlights the difficulties that arise.
Another interesting follow-up question is whether there is a general result that translates a Price of Anarchy guarantee for a given mechanism that is provable via smoothness into a result that shows that best-response sequences reach states of good social welfare quickly. The example with MPH-$k$ valuations in Appendix~\ref{app:mphk} already limits the potential scope of such a result. It would still be interesting to identify natural sufficient conditions. One such condition could be that the mechanism admits some kind of potential function (as the procedure for XOS valuations), but our results already show that this condition is certainly not necessary.



\bibliographystyle{abbrvnat}
\bibliography{best-response}


\appendix

\section{Non-Existence of Weak No-Overbidding Pure Nash Equilibria for Subadditive Valuations}\label{app:no-pne}
We can also leverage our novel insights regarding hard instances (Definition \ref{def:hard-instances}) for subadditive valuations to show that there need not be a pure Nash equilibrium, even if players are only required to use weakly no-overbidding strategies. 

\begin{theorem}
Let $k \in \mathbb{N}_{>0}$. Consider the hard instance $\mathcal{I}_k$ with $n = 2$ players and $m = 2^k-1$ items. There is no pure Nash equilibrium in weakly no-overbidding strategies if $k \geq 8$. This remains true if we define a bid profile to be at equilibrium if no player has a beneficial deviation to a weakly no-overbidding strategy.
\end{theorem}

\begin{proof}
Assume that $b$ is a weakly no-overbidding pure Nash equilibrium. Suppose the second player wins the set of items $W \subseteq M$ in $b$, then the first player wins the set of items $M \setminus W$. By weak no-overbidding, we have
\[
\sum_{j \in M \setminus W} b_{1, j} \leq v_1(M \setminus W) \;\; \text{and} \;\; \sum_{j \in W} b_{2, j} \leq v_2(W) \enspace.
\]

The first player does not win the items in $W$, which means that $b_{1, j} \leq b_{2, j}$ for all items $j \in W$. Consequently, we have
\begin{align*}
\sum_{j \in M} b_{1, j} 
&\leq v_1(M \setminus W) + v_2(W)\\
&\leq v_1(M) + v_2(M)\\ 
&= k + \rho \cdot 2^d \\
&= k + 4 \cdot \frac{k}{m} \cdot 2^{k - \log_2 k} \\
&= k + 4 \enspace.
\end{align*}

By the same argument as in Lemma \ref{lem:prices}, each item $j \in M$ is included in the same number of sets $D \in \mathcal{D}$. Therefore,
\[
\frac{1}{\lvert \mathcal{D} \rvert} \sum_{D \in \mathcal{D}} \frac{1}{\lvert D \rvert} \sum_{j \in D} b_{1, j} = \frac{1}{m} \sum_{j \in M} b_{1, j} \leq \frac{k + 4}{m} \enspace.
\]
This implies that there is a set $D \in \mathcal{D}$ such that
\[
\frac{1}{\lvert D \rvert} \sum_{j \in D} b_{1, j} \leq \frac{k + 4}{m} \enspace.
\]
Since $k \geq 8$ by assumption, $m > 2k + 8$, and therefore 
\[
\sum_{j \in D} b_{1, j} \leq \frac{k + 4}{m} \cdot \lvert D \rvert < \frac{\lvert D \rvert}{2} \enspace.
\]
By Lemma~\ref{lem:demand-set} and because the second player plays a best response, we have $W \supseteq D'$ for some $D' \in \mathcal{D}$.

In the remainder, we will show that this implies that the first player has a beneficial weakly no-overbidding deviation $b_1'$. 

Let $b_{1, j}' = b_{2, j} + \frac{1}{m}$ for $j \in W$ and $b_{1, j}' = b_{1, j}$ for $j \in M \setminus W$. Observe that in $(b_1', b_2)$ the first player wins all items $M$. This bid fulfills the weak no-overbidding property because
\begin{align*}
\sum_{j \in M} b_{1, j}' 
&= \sum_{j \in W} \left( b_{2, j} + \frac{1}{m} \right) + \sum_{j \in M \setminus W} b_{1, j}\\
&\leq v_2(W) + 1 + v_1(M \setminus W) \\
&\leq v_2(D') + 1 + v_1(M \setminus D') \\
&\leq \rho 2^d + 1 + k - d \\
&= 4 + 1 + \log_2 k \\
&\leq k \\
&= v_1(M) \enspace,
\end{align*}
where the first inequality uses that $b$ is weakly no-overbidding, the second inequality exploits the definition of $v_2$, the third inequality holds by Lemma \ref{lem:value}, and the final inequality holds because we have assumed $k \geq 8$.

The deviation by the first player is beneficial because 
\begin{align*}
u_1(b_1', b_2) 
&= v_1(M) - \sum_{j \in M} b_{2, j} \\
&= k - d - \sum_{j \in M \setminus W} b_{2, j} + d - \sum_{j \in W} b_{2, j}\\
&\geq u_1(b) + d - v_2(W) \\
&\geq u_1(b) + d - 4 > u_1(b) \enspace,
\end{align*}
where the first inequality uses Lemma \ref{lem:value}, the second inequality uses that $v_2(W) \leq v_2(D') = 4$, and the final inequality follows from the definition of $d = k - \log_2 k$ and the assumption that $k \geq 8$ and so $d > 4$.
\end{proof}

\section{Missing Proofs from Section~\ref{sec:prelims}}
\label{apx:prelims}
In this appendix we prove the propositions that establish the existence of aggressive and safe bidding dynamics for XOS and subadditive valations.

\subsection{Sufficiency of Strong No-Overbidding}
We first show that in order to have a $1$-safe dynamic it suffices that initial bids and the subsequent updates fulfill no-overbidding in the strong sense. 
A bid vector $b$ satisfies strong no-overbidding if $\sum_{j \in S} b_{i,j} \leq v_i(S)$ for every bidder $i$ and every set of items $S$. 
A best response $b_i$ by bidder $i$ against bids $b_{-i}$ satisfies strong no-overbidding if $\sum_{j \in S} b_{i,j} \leq v_i(S)$.

\begin{lemma}\label{lem:no-overbidding}
If the initial bid vector $b^0$ satisfies strong no-overbidding and at each time step $t \geq 1$ some bidder $i$ gets to update his bid to a best response, which satisfies  strong no-overbidding, then the resulting best-response dynamic is $1$-safe.
\end{lemma}
\begin{proof}
Since the initial bid vector and each update satisfy strong no-overbidding we have $\sum_{j \in S} b^t_{i,j} \leq v_i(S)$ for all bidders $i$, time steps $t \geq 0$, and sets of items $S$. Subtracting $\sum_{j \in S} \max_{k \neq i} b_{k,j}^t$ from both sides shows the claim.
\end{proof}

\subsection{Proof of Proposition \ref{prop:xos}}
Consider an arbitrary bidder $i$ and his update to bid $b_i^t$. The bid $b_i^t$ satisfies strong no-overbidding by definition. Hence Lemma \ref{lem:no-overbidding} shows that the bid sequence is $1$-safe. It remains to show that $b_i^t$ is a $1$-aggressive best response.

We first show that the bid $b_i^t$ is a best response to $b_{-i}^t$. Let $S_i$ denote the set of items that bidder $i$ wins with bid $b^t_i$ against bids $b^t_{-i}$ and let $D$ be the demand set on the basis of which $b_i^t$ is defined. Then, 
\begin{align*}
u_i(b^t) &= v_i(S_i) - \sum_{j \in S_i} \max_{k \neq i} b^t_{k,j}  
\geq \sum_{j \in S_i} (b^t_{i,j} - \max_{k \neq i} b^t_{k,j}) \\
&\geq \sum_{j \in D} (b^t_{i,j} - \max_{k \neq i} b^t_{k,j})
= v_i(D) - \sum_{j \in D} \max_{k \neq i} b^{t}_{k,j} \\
&\geq \max_{S} \bigg( v_i(S) - \sum_{j \in S} \max_{k \neq i} b^t_{k,j} \bigg),
\end{align*}
where the first inequality uses that $v_i$ is XOS, the second uses that $\max_{k \neq i} b^t_{k,j} = b^t_{i,j}$ for $j \in D \setminus S_i$ and $\max_{k \neq i} b^t_{k,j} \leq b^t_{i,j}$  for  $j \in S_i \setminus D$, the following equality exploits the definition of $b^t_i$, and the final inequality uses that $D$ is a demand set.

To show that $b_i^t$ is $1$-aggressive it suffices to show that bidder $i$'s declared and actual utility at time step $t$ coincide.  
Since the right-hand side in the preceding chain of inequalities is at least $v_i(S_i) - \sum_{j \in S_i} \max_{k \neq i} b^t_{k,j}$ all inequalities in the chain of inequalities must be equalities.
This implies that
\begin{align*}
	u_i(b^t) 
	= v_i(S_i) - \sum_{j \in S_i} \max_{k \neq i} b^t_{k,j} 
	= \sum_{j \in S_i} (b^t_{i,j} - \max_{k \neq i} b^t_{k,j}) 
	= u_i^D(b^t)\enspace.
\end{align*}

\subsection{Proof of Proposition \ref{prop:subadditive-1}}
Consider an arbitrary bidder $i$ and his update to bid $b_i^t$. We first argue that $b_i^t$ is a best response. 
We claim that $\tilde{u}_i(S,b_{-i}^t) > 0$ for all $S \subseteq D$ unless $D = \emptyset$. To see this assume by contradiction that there exist a $S \subseteq D$ such that $\tilde{u}_i(T,b_{-i}^t) \leq 0$. Then, by subadditivity of $v_i$,
\begin{align*}
	\tilde{u}_i(D,b_{-i}) 
	&\leq \bigg(v_i(D\setminus T) - \sum_{j \in D \setminus T} \max_{k \neq i} b_{k,j}\bigg) + \bigg(v_i(T) - \sum_{j \in S} \max_{k \neq i} b_{k,j}\bigg)\\
	&\leq \tilde{u}_i(D \setminus T,b_{-i}^t)\enspace,
\end{align*}
which contradicts the definition of $D$.
Because of this the additive approximation $a_{i}$ has $a_{i,j} > 0$ for all $j \in D$. It follows that $b_{i,j}^t > \max_{k \neq i} b_{k,j}^t$ for all $j \in D$, and so bidder $i$ wins all items $j \in D$, and for the items $j \not \in D$ that he wins $\max_{k \neq i} b_{k,j}^t = 0$.

To see that $b_i^t$ is $1/\ln m$-aggressive observe the following. Let $S_i$ denote the set of items that bidder $i$ wins with bid $b_i^t$. Then, considering the bid $b_i^t$ defined on the basis of demand set $D$, we have
\begin{align*}
u_i^D(b^t) 
&= \sum_{j \in S_i} \left(b_{i,j}^t - \max_{k \neq i} b_{k,j}^t\right) \\
&\geq \sum_{j \in D} \left(b_{i,j}^t - \max_{k \neq i} b_{k,j}^t\right) \\
&= \sum_{j \in D} a_{i,j} \geq \frac{1}{\ln m} \cdot \tilde{u}_i(D,b_{-i}^t)\enspace,
\end{align*}
where the first inequality uses that $b_{i,j}^t = \max_{k \neq i}b_{k,j}^t$ for $j \in D \setminus S_i$ and $b_{i,j}^t \geq \max_{k \neq i}b_{k,j}^t$ for $j \in S_i \setminus D$, and the second inequality uses property (a) of bid $b_i^t$.

That the bid sequence is $1$-safe follows from the starting condition and Lemma \ref{lem:no-overbidding} by observing that bidder $i$'s update satisfies strong no-overbidding. Namely, for every $S \subseteq D$,
\begin{align*}
	\sum_{j \in S} b_{i,j}^t 
	= \sum_{j \in S} (a_{i,j} + \max_{k \neq i}b_{k,j}^t) 
	\leq \tilde{u}_i(S,b_{-i}^t) + \sum_{j \in S} \max_{k \neq i}b_{k,j}^t 
	= v_i(S)\enspace,
\end{align*}
where the inequality follows from property (b) of bid $b_i^t$.

\subsection{Proof of Proposition \ref{prop:subadditive-2}}
The argument that the bid $b_i^t$ chosen by bidder $i$ is a best response and $1$-aggressive is identical to the respective argument in the proof of Proposition \ref{prop:subadditive-1}, except that this time we collect a factor of $1$ instead of $1/\ln m$ when we apply property (a) of bid $b_i^t$. 

To see that the bid sequence is $\ln m$-safe, consider a point in time $t' \geq t$ after bidder $i$'s update. In the vector $b_{t'}$, bidder $i$ gets a set $S \subseteq M$ that is possibly different from $D$. Note that for $j \in S \setminus D$, $b_{i, j}^{t'} = 0$ by our definition. Furthermore, for $j \in S \cap D$, $\max_{k \neq i} b_{k, j}^{t'} \leq \max_{k \neq i} b_{k, j}^t$ because bid updates are only non-zero if an item changes its owner. Therefore, because bidder $i$ wins item $j$, all new bids have to be zero.

In combination, we have
\begin{align*}
u_i^D(b^{t'}) 
&= \sum_{j \in S \cap D} \left(\tilde{a}_{i,j} + \max_{k \neq i} b_{k,j}^t - \max_{k \neq i} b_{k,j}^{t'} \right)\\
&\leq \ln m \cdot \Big(\tilde{u}_i(S \cap D,b_{-i}^t) + \sum_{j \in S \cap D} \big(\max_{k \neq i} b_{k,j}^t - \max_{k \neq i} b_{k,j}^{t'}\big)\Big)\\
&= \ln m \cdot u_i(b^{t'})\enspace,
\end{align*}
because the sum of $\tilde{a}_{i,j}$ terms is bounded by $\ln m \cdot \tilde{u}_i(S \cap D,b_{-i}^t)$ by definition and the sum of the remaining terms is non-negative.

\section{Tightness of the Point-Wise Welfare Guarantee for XOS Valuations}\label{app:one-third}
The following proposition shows that the point-wise welfare guarantee of $1/3$ for the round-robin best-response dynamics for fractionally subadditive valuations described in Section \ref{sec:prelims} is tight, even if the valuations are unit demand.

\begin{proposition}
Consider the dynamics described in Section \ref{sec:xos}.
There is an input with $n = 3$ players, $m = 3$ items, and unit-demand valuations and an initial bid vector such that when started from this bid vector the social welfare obtained by the dynamics after a single round of bid updates is $1/3 \cdot OPT(v)$.
\end{proposition}
\begin{proof}
The valuations of all three bidders are unit demand, i.e., for all players $i$ and sets of items $S$, $v_i(S) = \max_{j \in S} v_{i, j}$. The item valuations $v_{i,j}$ for $1 \leq i,j \leq 3$ are given by the following table:
\medskip
\begin{center}
\begin{tabular}{r|ccc}
& item 1 & item 2 & item 3 \\\hline
player 1 & $1$ & $0$ & $0$ \\
player 2 & $1 + \epsilon$ & $1 + 2 \epsilon$ & $1 + 3 \epsilon$ \\
player 3 & $0$ & $0$ & $1$
\end{tabular}
\end{center}
\medskip

Suppose that the XOS representation of these valuations is that each player has an additive valuation $a^{i,0}$ that is all zero and then one for each item $j$, $a^{i,j}$, such that $a^{i,j}(k) = v_{i,j}$ for $k = j$ and $a^{i,j}(k) = 0$ otherwise.

Let $b^0$ be the bid profile in which Player 2 bids $1 + \epsilon$ on item $1$, all other bids are $0$. That is, $b^0 = (a^{1,0},a^{2,1},a^{3,0})$. Suppose that the order of updates is first player $1$ gets to update his bid, then player $2$, and then player $3$.

Player $1$ is already playing a best response to $b^0_{-1}$, so $b^1 = b^0$. Now, to get $b^2$, player $2$ updates his bids to a best-response to $b^1_{-2}$, which is $a^{2,3}$. That is, he bids zero on the first two items and $1+3 \epsilon$ on the third. So $b^2 = (a^{1,0},a^{2,3},a^{3,0})$. With these bids, however, bidding $0$ on all items is a best-response of player 3, therefore $b^3 = b^2$.

Observe that $SW(b^3) = DW(b^3) = 1 + 3 \epsilon$, whereas the optimal social welfare is $3 + 2 \epsilon$. The claim follows by letting $\epsilon$ tend to zero.
\end{proof}

\section{Proof of Theorem \ref{thm:randomized-activation}}
\label{app:randomized-activation}
In this appendix we provide additional details for the proof of Theorem~\ref{thm:randomized-activation}.
We first prove Lemma~\ref{lem:max-vs-final}. Afterwards, we state and prove Lemma \ref{lem:max-low}.


\subsection{Proof of Lemma~\ref{lem:max-vs-final}}
For a fixed $T$, let $y_j = \max_{t \leq T} \max_i b_{i, j}^t$ and $p_j^t = \max_i b_{i, j}^t$ for $t \leq T$. We first show that for all $x > 0$
\begin{equation}
\label{eq:randomizedactivation:probabilityinequality}
\Pr{y_j \geq x} \leq \left( 1 - \frac{1}{n} \right)^{-T} \Pr{p_j^T \geq x}
\end{equation}

To show \eqref{eq:randomizedactivation:probabilityinequality}, we use that $y_j$ is defined to be $\max_{t' \leq T} p_j^{t'}$. That is, if $y_j \geq x$, there has to be a $t' \in \{ 0, 1, \ldots, T\}$ for which $p_j^1 < x, \ldots, p_j^{t' - 1} < x, p_j^{t'} \geq x$. Note that for different $t'$ these are disjoint events, so
\[
\Pr{y_j \geq x} = \sum_{t' = 0}^T \Pr{p_j^1 < x, \ldots, p_j^{t' - 1} < x, p_j^{t'} \geq x} \enspace.
\]
Let us fix $t'$ and consider the event that $p_j^1 < x, \ldots, p_j^{t' - 1} < x, p_j^{t'} \geq x$. If $t' > 0$, in step $t'$ a player $i$ has been selected that whose bid has set $p_j^{t'} \geq x$; if $t' = 0$, the initial bid of some player $i$ on item $j$ is at least $x$. We have have $p_j^T < x$ only if this player $i$ is selected to update his bid in steps $t' + 1, \ldots, T$. This happens with probability $1 - \left( 1 - \frac{1}{n} \right)^{T - t'} \leq 1 - \left( 1 - \frac{1}{n} \right)^T$. Formally, we have
\[
\Pr{p_j^T < x \growingmid p_j^1 < x, \ldots, p_j^{t' - 1} < x, p_j^{t'} \geq x} \leq 1 - \left( 1 - \frac{1}{n} \right)^{T} \enspace.
\]
This implies
\begin{multline*}
\Pr{p_j^T \geq x, p_j^1 < x, \ldots, p_j^{t' - 1} < x, p_j^{t'} \geq x}
\geq \left( 1 - \frac{1}{n} \right)^{T} \Pr{ p_j^1 < x, \ldots, p_j^{t' - 1} < x, p_j^{t'} \geq x} \enspace.
\end{multline*}
We thus obtain
\begin{align*}
\Pr{y_j \geq x} & = \sum_{t' = 0}^T \Pr{p_j^1 < x, \ldots, p_j^{t' - 1} < x, p_j^{t'} \geq x} \\
& \leq \left( 1 - \frac{1}{n} \right)^{-T} \sum_{t' = 0}^T \Pr{p_j^T \geq x, p_j^1 < x, \ldots, p_j^{t' - 1} < x, p_j^{t'} \geq x} \\
& = \left( 1 - \frac{1}{n} \right)^{-T} \Pr{p_j^T \geq x} \enspace.
\end{align*}
This concludes the proof of \eqref{eq:randomizedactivation:probabilityinequality}.

To show the lemma, let $\epsilon > 0$. We use that the expectation of a non-negative random variable $X$ can be approximated by $\sum_{k=0}^\infty \epsilon \cdot \Pr{X \geq k \cdot\epsilon} \leq \Ex{X} \leq \sum_{k=1}^\infty \epsilon \cdot \Pr{X \geq k \cdot \epsilon}$. Applying this approximation and using \eqref{eq:randomizedactivation:probabilityinequality}, we get
\begin{align*}
\Ex{p_j^T} 
& \geq \sum_{k=1}^\infty \epsilon \Pr{p_j^T \geq k \epsilon}\\
&\geq \sum_{k=0}^\infty \left( 1 - \frac{1}{n} \right)^{T} \epsilon \Pr{y_j \geq k \epsilon} - \epsilon\\
&\geq \left( 1 - \frac{1}{n} \right)^{T} \Ex{y_j} - \epsilon \enspace.
\end{align*}
As this holds for all $\epsilon > 0$, we also have
\[
\Ex{p_j^T} \geq \left( 1 - \frac{1}{n} \right)^{T} \Ex{y_j} \enspace. 
\]


\subsection{Lemma \ref{lem:max-low} and Its Proof}
Next we state and prove Lemma~\ref{lem:max-low}, which we used in the proof of Theorem~\ref{thm:randomized-activation}.
\begin{lemma}\label{lem:max-low}
Let $S^\ast_1, \dots, S^\ast_n$ be any feasible allocation, in which player $i$ receives items $S^\ast_i$. Consider a sequence $b^0, \ldots, b^T$ in which each player from $N'$ updates his bid at least once using an $\alpha$-aggressive bid. We have
 \(
 	(\alpha + 1) \cdot DW(b^T) + \alpha \cdot \sum_{j \in M} \max_{t \leq T} \max_i b_{i, j}^t \geq \alpha \cdot \sum_{i \in N'} v_i(S^\ast_i).
 \)
 \end{lemma}

To prove this lemma we need the following auxiliary lemma.

\begin{lemma}\label{lem:aux-variant}
Consider a sequence $b^0, \ldots, b^T$ in which bidders from $N'$ update their bid at least once. For $i \in N'$, let $t_{i}$ denote the time of the last update for bidder $i$. Then,
\(
	\sum_{i \in N'} u_i^D(b^{t_i}) \leq DW(b^T).
\)
\end{lemma}
\begin{proof}
Without loss of generality, let $N' = \{1, \ldots, n'\}$ and $t_1 < t_2 < \ldots < t_{n'}$. Consider any $i \in N'$ and let bidder $i$'s update buy him the set of items $S'$. Then
\[
	u_i^D(b^{t_i}) = \sum_{j \in S'} \left(b^{t_i}_{i,j} - \max_{k \neq i} b_{k,j}^{t_i} \right)\enspace.
\]

For $i \in N'$, let $z^i_j = \max_{k < i} b_{k, j}^{t_i}$ for all $j$, $z^0_j = 0$. That is, $z^i_j$ is the highest ``final'' bid on item $j$.

We observe that
\[
	\sum_{j \in S'} (b_{i,j}^{t_i} - \max_{k \neq i} b_{k,j}^{t_i}) \leq \sum_{j \in M} (z_j^i - z_j^{i-1}) \enspace.
\]
This is for the following fact. For $j \not\in S'$, we have $z_j^i \geq z_j^{i-1}$ by definition. For $j \in S'$, $b_{i, j}^{t_i} = z^i_j$ and $\max_{k \neq i} b_{k,j}^{t_i} \geq \max_{k < i} b_{k,j}^{t_i} = \max_{k < i} b_{k,j}^{t_{i-1}} = z_j^{i-1}$.

By summing over all bidders $i \in N'$, we obtain
\[
	\sum_{i \in N'} u_i^D(b^{t_i}) \leq \sum_{i \in N'} \sum_{j \in M} (z_j^i - z_j^{i-1}).
\]

The double sum is telescoping and $z_j^T = z_j^{t_{n'}} = \max_{k \leq n'} b^T_{k,j} \leq \max_{k} b^T_{k,j}$ and $z_j^0 = 0$ by definition. So,
\[
	\sum_{i \in N'} u_i^D(b^{t_i}) \leq \sum_{j \in M} (z_j^T - z_j^{0}) = \sum_{j \in M} \max_{k} b^T_{k,j} = DW(b^T)\enspace. \qedhere
\]
\end{proof}

We are now ready to prove the lemma.

\begin{proof}[Proof of Lemma~\ref{lem:max-low}]
For $i \in N'$, let $t_i$ denote the last time player $i$ updates his bid. Instead of choosing bid $b^{t_i}_i$, he could have bought the set of items $S^\ast_i$. As $b^{t_i}_i$ is $\alpha$-aggressive, we get
\[
	u_i^D(b^{t_i}) \geq \alpha \cdot \bigg( v_i(S^\ast_i) - \sum_{j \in S^\ast_i} \max_{k \neq i} b_{k,j}^{t_i} \bigg) \enspace.
\]
Let $y_j = \max_t \max_k b_{k, j}^t$.

We thus have
\[
	u_i^D(b^{t_i}) + \alpha \cdot \sum_{j \in S^\ast_i} y_j \geq \alpha \cdot v_i(S^\ast_i)\enspace.
\]

Summing this inequality over all bidders $i \in N'$ yields
\[
	\sum_{i \in N'} u_i^D(b^{t_i}) + \alpha \cdot \sum_{i \in N'} \sum_{j \in S^\ast_i} y_j \geq \alpha \cdot \sum_{i \in N'} v_i(S^\ast_i) \enspace.
\]

The first sum is at most $DW(b^T)$ by Lemma \ref{lem:aux-variant}. The double sum covers each $j \in M$ at most once, therefore it is bounded by $\sum_{j \in M} y_j$. Consequently,
\[
	DW(b^T) + \alpha \cdot \sum_{j \in M} y_j \geq \alpha \cdot \sum_{i  \in N'} v_i(S^\ast_i) \enspace. \qedhere
\]
\end{proof}

\section{Proof of Theorem~\ref{thm:adversarial-activation}}\label{app:adversarial-activation}
Our proof of Theorem~\ref{thm:adversarial-activation} combines the construction that we used to prove Proposition~\ref{prop:adversarial-activation} with the following exponential lower-bound construction.


\begin{lemma}[Theorem 3.4~of~\cite{ChristodoulouKS16}]\label{lem:exponential}
For every $k$ there is an instance with two players, $A$ and $B$, and $k$ items, with fractionally subadditive valuations $v_A$ and $v_B$ defined by additive functions $(a_A^t)_{t \in \mathbb{N}}$ and $(a_B^t)_{t \in \mathbb{N}}$ such that in the Potential Procedure, started from initial bid vector $b^0$ in which both players bid zero and with player $A$ making the first move, player $z \in \{A,B\}$ plays $a_z^t$ the $t$-th time he gets to update his bid and it takes at least $\Omega(2^k)$ steps before the procedure converges. 
\end{lemma}


\begin{proof}[Proof of Theorem \ref{thm:adversarial-activation}]
As in the proof of Proposition \ref{prop:adversarial-activation} we use $n$ players, we start with the initial bid vector $b^0$ in which all players bid zero, and we consider player $1$ being activated in every odd step and the remaining players being activated in round-robin fashion in even steps.

We use $m = (n-1) \cdot (k+1)$ items. Items $1, \ldots, n-1$ are used to mimic the sequence of Proposition~\ref{prop:adversarial-activation}. The remaining items are grouped into $n-1$ sets of size $k$, namely $C_i := \{ n - 1 + (i - 2) k + 1, \ldots, n - 1 + (i - 1) k \}$ for $i > 2$, and on each of these sets player $1$ follows the steps of the exponential-length sequence of Lemma~\ref{lem:exponential} with one of the other $n-1$ players, with player $1$ taking the role of player $A$ and player $i > 1$ taking the role of player $B$.

To define the valuations, for $z \in \{A, B\}$, $i = 2, \ldots, n$, and $t \geq 1$, let $a_{z, i}^t$ be the additive valuation functions defined in Lemma \ref{lem:exponential} that are used by player $z \in \{A,B\}$ after the $t$-th update, using the items $C_i$.

We first define the valuation function $v_i$ for players $i > 1$. Namely, given some $\epsilon > 0$, let the valuation function $v_i$ of player $i > 1$ be defined as
\[
v_i(S) = \max\{ \mathbf{1}_{i-1 \in S}, \epsilon \cdot \max_t a_{B, i}^t(S) \} \enspace.
\]
That is, player $i$ has a high value to buy item $i-1$. He also has a very small value for items $C_i$ according to the valuations of player $B$ in the exponential lower-bound construction using the items $C_i$. 

For player $1$, we define the valuation function by setting $v_1(S) = \max_t v_1^t(S)$, where $v_1^t$ is the additive valuation function that is used when player $1$ updates his bid for the $t$-th time. It is designed in such a way that the $t$-th update is a best response in the game on $C_i$ with player $i = (t - 1) \mod (n-1) + 1$, who has just updated his bid, and makes the bid of bidder $1$ move from item $i - 1$ to $i$, which bidder $i + 1$ is interested in, who will be activated next.

To define $v_1^t$ formally, observe that when player $1$ makes his $t$-th update, some of the other players have performed $\lceil \frac{t}{n - 1} \rceil$ updates so far, the others only $\lfloor \frac{t}{n - 1} \rfloor$. Let the respective sets of players be denoted by $N'(t)$ and $N''(t)$. Based on this, define
\[
v_1^t(S) = (1 + \epsilon) \cdot \mathbf{1}_{(t - 1) \bmod (n-1) + 1 \in S} + \epsilon \cdot \sum_{i \in N'(t)} a_{A, i}^{\lceil \frac{t}{n - 1} \rceil}(S) + \epsilon \cdot \sum_{i \in N''(t)} a_{A, i}^{\lfloor \frac{t}{n - 1} \rfloor}(S) \enspace.
\]

By these definitions, the bids on items $1, \ldots, n-1$ change exactly the way as in the proof of Proposition~\ref{prop:adversarial-activation} as long as there are still changes on items $C_i$ for $i > 1$. By Lemma \ref{lem:exponential} it takes at least $\Omega(2^k)$ updates until such a set $C_i$ reaches a stable state. Therefore, our constructed best-response sequence has low welfare at least until every player $2, \ldots, n$ has updated his bid at least $\Omega(2^k)$ times. Moreover, every update is the unique best response.
\end{proof}

\section{Negative Result for MPH-$k$ Valuations}\label{app:mphk}
The maximum over positive hypergraph-$k$ or MPH-$k$ hierarchy \cite{FeigeFIILS15} comprises valuation functions with different degrees of complementarity, as parametrized by $k$.
A valuation function $v_i$ belongs to MPH-$k$ if there are values $v_{i, T}^\ell \geq 0$ such that $v_i(S) = \max_\ell \sum_{T \subseteq S, \lvert T \rvert \leq k} v_{i, T}^\ell$. Any (monotone) valuation function can be captured with $k= m$. Fractionally subadditive valuations are precisely the case $k = 1$.

Observe that for a usual valuation function even in MPH-$2$, the only bids that fulfill strong no-overbidding are zero on every item. Therefore, it is not possible that bidders bid $\alpha$-aggressively for $\alpha > 0$ and satisfy no-overbidding in the strong sense at the same time. However, as our dynamics in Section~\ref{sec:subadditive-2} demonstrates, strong no-overbidding is not a necessary requirement for good welfare guarantees. Unfortunately, the case is different for MPH-$k$. Below we show a negative result for the valuation class MPH-$3$. It relies on ties regarding identical bids and multiple best responses being broken to the disadvantage of the dynamics.

\begin{proposition}
There are valuation functions for $n$ bidders on $O(n)$ items that belong to MPH-$3$ such that round-robin best-response dynamics only reach states that achieve a $O(\frac{1}{n})$-fraction of the optimal social welfare.
\end{proposition}

\begin{proof}
For a given $k$, we define an instance with $k + 4$ items and $2k+4$ bidders as follows. Bidder $i \in [k - 1]$ has a valuation of $3$ for the bundles $\{i, k+1, k+2\}$ and $\{i, k+3, k+4\}$, with no value for the subsets. Bidder $k$ has a valuation of $3$ for the bundles $\{k, k+1, k+3\}$ and $\{k, k+2, k+4\}$, with no value for the subsets. Furthermore, there are $k+4$ bidders $k+1, \ldots, 2k+4$, each of which has a valuation of $1$ for exactly one (distinct) item $j \in [k+4]$. Note that due to bidders $k+1, \ldots, 2k+4$, the optimal social welfare is $k+4$. Our best-response sequence will never reach a state with social welfare higher than $3$.

We assume that ties are broken as follows. Bidders $k+1, \ldots, 2k+4$ never get an item if there is an equal bid from a bidder $i \in [k]$. Among the bidders $i \in [k]$, on items $k+1$ and $k+3$, bidder $k$ is preferred to $k-1$, bidder $k-1$ to $k-2$, and so on. On items $k+2$ and $k+4$, bidders $i \in [k-1]$ are preferred to bidder $k$, bidder $k-1$ is preferred to $k-2$, bidder $k-2$ to $k-3$, and so on.

Now consider the round-robin best-response dynamics in which bidders get activated in the order they are indexed. Throughout the bidding dynamics bidders $k+1, \ldots, 2k+4$ will bid truthfully on their respective items. The other bidders bid as follows. In odd rounds bidders $i = 1, \ldots, k-1$ buy items $\{i, k+1, k+2\}$, bidding $1$ on each of them. Afterwards, bidder $k$ buys items $\{k, k+1, k+3\}$, again bidding $1$ on each of them.
In even rounds, bidders $i = 1, \ldots, k-1$ buy items $\{i, k+3, k+4\}$, bidding $1$ each, making bidder $k$ buy items $\{k, k+2, k+4\}$. 

Note that at every point in this sequence, only the bidder that has just updated his bid gets a bundle of items of any positive value. This value is $3$.
\end{proof}

\section{Lazy Updates}\label{app:lazy}

In this appendix we show that our results also transfer to the case in which updates are lazy. That is, a bidder may also choose not to update the bids when he is already playing a best response given the current other bids. It is now important to assume that bid updates are zero for items that are not won and that no item is ever won with bid zero. We will consider the points in time when each bidder has performed at least one $\alpha$-aggressive update.

\begin{theorem}
\label{thm:main-1-lazy}
In a $\beta$-safe round-robin bidding dynamic with lazy $\alpha$-aggressive bid updates the social welfare at any time step $t$  after which each bidder has performed at least one $\alpha$-aggressive update satisfies
\(
	SW(b^t) \geq \frac{\alpha}{(1+2\alpha+\beta)\beta} \cdot OPT(v).
\)
\end{theorem}

To prove this theorem, we use variants of Lemmas \ref{lem:initial-low} and \ref{lem:initial-high} that do not rely on eager updates. Note that Lemma~\ref{lem:declared-vs-actual} does not rely on $\alpha$-aggressive updates and therefore continues to hold in the case of lazy updates. Our first lemma generalizes Lemma~\ref{lem:initial-low}.

\begin{lemma}
\label{lem:initial-low-lazy}
Let $S^\ast_1, \dots, S^\ast_n$ be any feasible allocation, in which player $i$ receives items $S^\ast_i$. Consider a round-robin sequence $b^0, \ldots, b^T$ in which each player updates his bid at least once using an $\alpha$-aggressive bid and may be lazy afterwards. We have
 \(
 	(2 \alpha + 1) \cdot DW(b^T) + \alpha \cdot DW(b^{T-n}) \geq \alpha \cdot \sum_{i \in N'} v_i(S^\ast_i).
 \)
 \end{lemma}
 
 \begin{proof}
Let $t_i$ denote the last time player $i$ updates his bid and $t_i'$ denote the last time he is offered to update the bid. Let the set bought at time $t_i$ be $S_i$, the set that is still won at time $t_i'$ be $S'_i \subseteq S_i$. Instead of choosing bid $b^{t_i}_i$, he could have bought the set of items $S'$. As $b^{t_i}_i$ is $\alpha$-aggressive, we get
\[
	u_i^D(b^{t_i}) \geq \alpha \cdot \bigg( v_i(S'_i) - \sum_{j \in S'_i} \max_{k \neq i} b_{k,j}^{t_i} \bigg) \enspace.
\]
The declared utility at time $t_i'$ is given by
\[
	u_i^D(b^{t_i'}) = \sum_{j \in S'_i} \left( b_{i, j}^{t_i} - \max_{k \neq i} b_{k,j}^{t_i'} \right) \enspace.
\]
In combination, we get
\begin{align*}
	u_i^D(b^{t_i}) + \alpha u_i^D(b^{t_i'}) & \geq \alpha \cdot \bigg( v_i(S'_i) - \sum_{j \in S'_i} \max_{k \neq i} b_{k,j}^{t_i} + \sum_{j \in S'_i} b_{i, j}^{t_i} - \sum_{j \in S'_i} \max_{k \neq i} b_{k,j}^{t_i'} \bigg) \\
	& = \alpha u_i(b^{t_i'}) + \alpha \cdot \bigg( \sum_{j \in S'_i} b_{i, j}^{t_i} - \sum_{j \in S'_i} \max_{k \neq i} b_{k,j}^{t_i} \bigg) \geq \alpha u_i(b^{t_i'}) \enspace,
\end{align*}
where in the last step we use that for every $j \in S_i' \subseteq S_i$ the update sets $b_{i, j}^{t_i} \geq \max_{k \neq i} b_{k,j}^{t_i}$.

At $t_i'$, bidder $i$ could buy the set $S^\ast_i$ instead. Therefore
\[
	u_i(b^{t_i'}) \geq v_i(S^\ast_i) - \sum_{j \in S^\ast_i} \max_{k \neq i} b_{k,j}^{t_i'} \enspace.
\]

We set $p_j^t = \max_i b_{i, j}^t$. As $T - n + 1\leq t_i' \leq T$, we have $p_j^{t_i'} \leq p_j^T + p_j^{T-n}$ by the same argument as in the proof of Lemma~\ref{lem:initial-low} and thus
\[
	u_i^D(b^{t_i}) + \alpha u_i^D(b^{t_i'}) + \alpha \cdot \sum_{j \in S^\ast_i} (p_j^T + p_j^{T-n}) \geq \alpha \cdot v_i(S^\ast_i)\enspace.
\]

Summing this inequality over all bidders $i \in N$ yields
\[
	\sum_{i \in N} (u_i^D(b^{t_i}) + \alpha u_i^D(b^{t_i'})) + \alpha \cdot \sum_{i \in N} \sum_{j \in S^\ast_i} (p_j^T + p_j^{T-n}) \geq \alpha \cdot \sum_{i \in N} v_i(S^\ast_i) \enspace.
\]

The first sum is at most $(1 + \alpha) DW(b^T)$ by Lemma \ref{lem:aux-variant}. The double sum covers each $j \in M$ at most once, therefore it is bounded by $DW(b^{T-n}) + DW(b^T)$. Consequently
\[
	(1 + \alpha) DW(b^n) + \alpha \cdot (DW(b^{T-n}) + DW(b^T)) \geq \alpha \cdot \sum_{i  \in N} v_i(S^\ast_i) \enspace. \qedhere
\]
\end{proof}

The second lemma generalizes Lemma~\ref{lem:initial-high}.

\begin{lemma}\label{lem:initial-high-lazy}
Consider a $\beta$-safe bid sequence $b^0, \ldots, b^n$ in which player $i$ changes his bid from $b^{i-1}$ to $b^i$ using an $\alpha$-aggressive bid or keeps it unchanged. Then, $DW(b^T) \geq \frac{\alpha}{\beta} \cdot DW(b^{T-n})$.
\end{lemma}

\begin{proof}
Without loss of generality, we assume that $T = n$. Otherwise shift the indices accordingly.

Consider an arbitrary bidder $i$ and his update from $b^{i-1}$ to $b^{i}$. We claim that
\begin{equation}
\label{eq:initial-high-lazy:changeofdw}
DW(b^i) \geq DW(b^{i-1}) - \left(\frac{\beta}{\alpha}-1\right) \cdot u^D_i(b^i) \enspace.
\end{equation}

Observe that if bidder $i$ keeps his bid unchanged, $DW(b^i) = DW(b^{i-1})$ and therefore \eqref{eq:initial-high-lazy:changeofdw} holds trivially. So, let us consider the case that bidder $i$ updates the bid $\alpha$-aggressively. Denote the set of items that bidder $i$ won under bids $b^{i-1}$ by $S_i^{i-1}$, and the set of items that he wins under bids $b^{i}$ by $S_i^{i}$. So 
\[
u_i^D(b^{i-1}) = \sum_{j \in S_i^{i-1}} b_{i,j}^{i-1} - \sum_{j \in S_i^{i-1}} \max_{k \neq i} b_{k,j}^{i-1} \quad \text{and} \quad
u_i^D(b^{i}) = \sum_{j \in S_i^{i}} b_{i,j}^{i} - \sum_{j \in S_i^{i}} \max_{k \neq i} b_{k,j}^{i}\enspace.
\]

Using that for all $k \neq i$ and all $j$ we have $b_{k,j}^{i-1} = b_{k,j}^{i}$ we obtain that the difference in declared welfare over all bidders between steps $i-1$ and $i$ is equal to the difference in bidder $i$'s declared utility at these time steps. Formally,
\begin{align*}
DW(b^i) 
&= \sum_{j \in M\setminus S_i^{i}} \max_{k \neq i} b_{k,j}^{i-1} + \sum_{j \in S_i^{i}} b_{i,j}^{i}\\
&=  \sum_{j \in M} \max_{k \neq i} b_{k,j}^{i-1} + \sum_{j \in S_i^{i}} b_{i,j}^{i} - \sum_{j \in S_i^{i}} \max_{k \neq i} b^{i}_{k,j}\\
&=  \sum_{j \in M} \max_{k \neq i} b_{k,j}^{i-1} + u_i^D(b_i) \\
&=  \sum_{j \in M\setminus S_i^{i-1}} \max_{k \neq i} b_{k,j}^{i-1} + \sum_{j \in S_i^{i-1}} \max_{k \neq i} b_{k,j}^{i-1} + u_i^D(b_i)\\
&=  \sum_{j \in M\setminus S_i^{i-1}} \max_{k \neq i} b_{k,j}^{i-1} + \sum_{j \in S_i^{i-1}} b_{i,j}^{i-1} + u_i^D(b_i) - \sum_{j \in S_i^{i-1}} b_{i,j}^{i-1} + \sum_{j \in S_i^{i-1}} \max_{k \neq i} b_{k,j}^{i-1}\\
&= DW(b^{i-1}) + u_i^D(b_i) - u_i^D(b^{i-1})\enspace.
\end{align*}

Since $b^i_i$ is $\alpha$-aggressive, 
we have $u^D_i(b^i) \geq \alpha \cdot u_i(b^{i-1})$. 
Since the bidding sequence is $\beta$-safe, $u^D_i(b^t) \leq \beta \cdot u_i(b^t)$ for all $t$. So,
\begin{align*}
DW(b^i) & = DW(b^{i-1}) + u^D_i(b^i) - u^D_i(b^{i-1}) \\
& \geq DW(b^{i-1}) + u^D_i(b^i) - \beta \cdot u_i(b^{i-1}) \\
& \geq DW(b^{i-1}) + u^D_i(b^i) - \frac{\beta}{\alpha} \cdot u^D_i(b^i) \\
& = DW(b^{i-1}) - \left(\frac{\beta}{\alpha}-1\right) \cdot u^D_i(b^i) \enspace.
\end{align*}
This implies that \eqref{eq:initial-high-lazy:changeofdw} also holds in this case.

Summing \eqref{eq:initial-high-lazy:changeofdw} over all bidders $i \in N$ and using the telescoping sum $\sum_{i \in N} (DW(b^i) - DW(b^{i-1}) = DW(b^n) - DW(b^0)$ we obtain
\[
DW(b^n) \geq DW(b^0)- \left(\frac{\beta}{\alpha}-1\right) \sum_{i \in N} u_i^D(b^i)\enspace.
\]
Since $\alpha \leq 1$ and $\beta \geq 1$ the factor $(\beta/\alpha-1) \geq 0$. We can therefore use Lemma \ref{lem:aux} to conclude that
\[
DW(b^n) \geq DW(b^0)- \left(\frac{\beta}{\alpha}-1\right) DW(b^n) \enspace.
\]
This implies the claim.
\end{proof}

\begin{proof}[Proof of Theorem \ref{thm:main-1-lazy}] 
Combining Lemma \ref{lem:initial-low-lazy} with Lemma \ref{lem:initial-high-lazy} to the allocation $S^*_1, \dots, S^*_n$ that maximizes social welfare we obtain
\begin{align*}
(1+2\alpha+\beta) \cdot DW(b^t) 
&= (2\alpha+1) \cdot DW(b^t) + \alpha \cdot \frac{\beta}{\alpha} DW(b^t)\\
&\geq (2\alpha+1) \cdot DW(b^t) + \alpha \cdot DW(b^{t-n-1}) \geq \alpha \cdot OPT(v).
\end{align*}

Now, by Lemma \ref{lem:declared-vs-actual}, $DW(b^t) \leq \beta \cdot SW(b^t)$. Combining this with the previous inequality yields
\[
	(1+2\alpha+\beta) \cdot \beta \cdot SW(b^t) \geq \alpha \cdot OPT(v). \qedhere
\]
\end{proof}

\end{document}